\documentclass[conference]{IEEEtran}
\date{}

\usepackage{fancyhdr}
\usepackage[normalem]{ulem}
\usepackage{flushend}
\usepackage{tikz}
\usepackage{amssymb}
\usepackage{amsmath}
\usepackage{amsthm}
\usepackage{xspace}
\usepackage{xcolor,colortbl}
\usepackage[font=footnotesize]{caption}
\usepackage{subcaption}
\usepackage[utf8]{inputenc}
\usepackage[linesnumbered,noend]{algorithm2e}
\SetAlgoSkip{}
\SetAlgoInsideSkip{}
\SetAlCapSkip{0.5em}
\SetAlCapNameFnt{\footnotesize}
\SetAlCapFnt{\footnotesize}
\usepackage{thmtools}
\usepackage{thm-restate}
\usepackage[hyphens,spaces,obeyspaces]{url}
\usepackage{hyperref}
\usepackage{cleveref}
\usepackage{bbding}
\usepackage{pifont}
\usepackage{wasysym}
\usepackage{multirow}
\usepackage{booktabs}
\usepackage{enumitem}
\usepackage{soul}
\usepackage{latexsym}
\usepackage{microtype}
\usepackage[textsize=footnotesize]{todonotes}  % TODO remove for release

\author{
\IEEEauthorblockN{Nicolas Derumigny\IEEEauthorrefmark{1}\IEEEauthorrefmark{3},
Théophile Bastian\IEEEauthorrefmark{1},
Fabian Gruber\IEEEauthorrefmark{1},
Guillaume Iooss\IEEEauthorrefmark{1}\\
Christophe Guillon\IEEEauthorrefmark{2},
Louis-No{\"e}l Pouchet\IEEEauthorrefmark{3},
Fabrice Rastello\IEEEauthorrefmark{1}
}

\IEEEauthorblockA{\IEEEauthorrefmark{1}
   	Univ. Grenoble Alpes, Inria, CNRS, Grenoble INP, LIG,
    38000 Grenoble, France
}

\IEEEauthorblockA{\IEEEauthorrefmark{2}
    STMicroelectronics, France
}

\IEEEauthorblockA{\IEEEauthorrefmark{3}
    Colorado State University,
    Fort Collins, Colorado, USA
}
}

\SetKwProg{Fn}{Function}{}{end}
\SetKwProg{Solve}{Solve}{}{}
\SetKwRepeat{DoUntil}{do}{until}

\newcommand\ipc[1]{{\mathop{\overline{#1}}}}

\def\Q{\mathbb{Q}}

\def\muop{$\mu$OP\xspace}
\def\muops{$\mu$OPs\xspace}
\def\basicInst{\mathcal{I}_{\mathcal{B}}}

\def\load{\normalfont\textrm{load}\xspace}

\def\Kernels{\mathcal{K}}
\def\Resources{\mathcal{R}}
\def\Instructions{\mathcal{I}}

\definecolor{resourcegray}{rgb}{0.4, 0.4, 0.4}
\newcommand{\resourcelabel}[1]{{\small\textcolor{resourcegray}{(#1)}}}

\definecolor{darkgreen}{rgb}{0.0, 0.5, 0.13}
\def\yes{\textcolor{darkgreen}{\ding{51}}}
\def\no{\textcolor{red}{\ding{55}}}

\newcommand{\na}{{\color{gray}N/A}}
\newcommand{\centerheader}[1]{\multicolumn{1}{c}{#1}}

\newcommand{\figfiverow}[2]{
        \begin{minipage}[c]{0.03\linewidth}
            \figrowlegend{#1}
        \end{minipage}\figspaceleft{}
        \begin{minipage}[c]{0.13\linewidth}
            \includegraphics[width=\linewidth]{images/ipc_heatmaps/#2-palmed.png}
        \end{minipage}\figspacemid{}
        \begin{minipage}[c]{0.13\linewidth}
            \includegraphics[width=\linewidth]{images/ipc_heatmaps/#2-UOPS.png}
        \end{minipage}\figspacemid{}
        \begin{minipage}[c]{0.13\linewidth}
            \includegraphics[width=\linewidth]{images/ipc_heatmaps/#2-PMEvo.png}
        \end{minipage}\figspacemid{}
        \begin{minipage}[c]{0.13\linewidth}
            \includegraphics[width=\linewidth]{images/ipc_heatmaps/#2-IACA.png}
        \end{minipage}\figspacemid{}
        \begin{minipage}[c]{0.13\linewidth}
            \includegraphics[width=\linewidth]{images/ipc_heatmaps/#2-LLVM_MCA.png}
        \end{minipage}\figspaceright{}
        \begin{minipage}[c]{0.035\linewidth}
            \includegraphics[width=\linewidth]{images/ipc_heatmaps/#2-legend.png}
        \end{minipage}
    }
\newcommand{\figthreerow}[2]{
        \begin{minipage}[c]{0.03\linewidth}
            \figrowlegend{#1}
        \end{minipage}\figspaceleft{}
        \begin{minipage}[c]{0.15\linewidth}
            \includegraphics[width=\linewidth]{images/ipc_heatmaps/#2-palmed.png}
        \end{minipage}\figspacezenmid{}
        \begin{minipage}[c]{0.15\linewidth}
            \includegraphics[width=\linewidth]{images/ipc_heatmaps/#2-PMEvo.png}
        \end{minipage}\figspacezenmid{}
        \begin{minipage}[c]{0.15\linewidth}
            \includegraphics[width=\linewidth]{images/ipc_heatmaps/#2-LLVM_MCA.png}
        \end{minipage}\figspaceright{}
        \begin{minipage}[c]{0.035\linewidth}
            \includegraphics[width=\linewidth]{images/ipc_heatmaps/#2-legend.png}
        \end{minipage}
    }
\newcommand{\figthreerowlegend}[2]{
        \begin{minipage}[c]{0.03\linewidth}
            \figrowlegend{#1}
        \end{minipage}\figspaceleft{}
        \begin{minipage}[c]{0.15\linewidth}
            \includegraphics[width=\linewidth]{images/ipc_heatmaps/#2-palmed.png}
        \end{minipage}\figspacezenmid{}
        \begin{minipage}[c]{0.12\linewidth}
            \includegraphics[width=\linewidth]{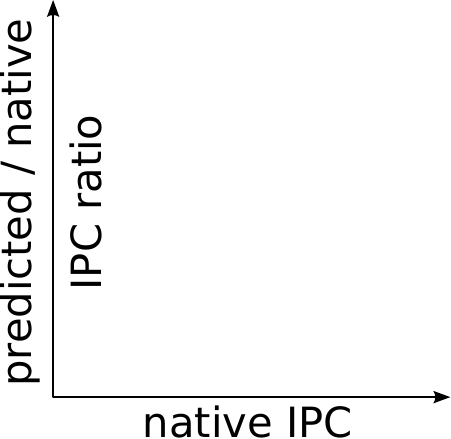}
        \end{minipage}\figspacezenmid{}
        \begin{minipage}[c]{0.15\linewidth}
            \includegraphics[width=\linewidth]{images/ipc_heatmaps/#2-PMEvo.png}
        \end{minipage}\figspacezenmid{}
        \begin{minipage}[c]{0.12\linewidth}
            ~
        \end{minipage}\figspacezenmid{}
        \begin{minipage}[c]{0.15\linewidth}
            \includegraphics[width=\linewidth]{images/ipc_heatmaps/#2-LLVM_MCA.png}
        \end{minipage}\figspaceright{}
        \begin{minipage}[c]{0.035\linewidth}
            \includegraphics[width=\linewidth]{images/ipc_heatmaps/#2-legend.png}
        \end{minipage}
    }
\newenvironment{figleftlabel}[1]{
        \begin{minipage}[c]{0.03\linewidth}
            \figrowlegend{#1}
        \end{minipage}\figspaceleft{}
        \begin{minipage}[c]{0.92\linewidth}}
    {\end{minipage}}

\newcommand{\figspaceleft}{\hspace{1em}}
\newcommand{\figspaceright}{\hspace{1em}}
\newcommand{\figspacemid}{\hfill}
\newcommand{\figspacezenmid}{\hfill}
\newcommand{\figcollegend}[1]{\begin{center}\textbf{\large #1}\end{center}}
\newcommand{\figrowlegend}[1]{\rotatebox{90}{\large\textbf{#1}}}
\newcommand{\figspacerow}

\newcommand{\tabcoverage}{\centerheader{Cov.}}
\newcommand{\taberror}{\centerheader{Err.}}
\newcommand{\tabkendall}{\centerheader{$\tau_K$}}
\newcommand{\tabcoverageunit}{\centerheader{(\%)}}
\newcommand{\taberrorunit}{\centerheader{(\%)}}
\newcommand{\tabkendallunit}{\centerheader{$(1)$}}
\definecolor{tabcovbg}{RGB}{235,255,255}
\definecolor{taberrbg}{RGB}{255,240,255}
\definecolor{tabktaubg}{RGB}{239,250,235}
\newcolumntype{C}{>{\columncolor{tabcovbg}}r}
\newcolumntype{E}{>{\columncolor{taberrbg}}r}
\newcolumntype{K}{>{\columncolor{tabktaubg}}r}

\def\bsr{\texttt{BSR}\xspace}

\def\divps{\texttt{DIVPS}\xspace}

\def\jmp{\texttt{JMP}\xspace}

\def\jnle{\texttt{JNLE}\xspace}

\def\addss{\texttt{ADDSS}\xspace}

\def\vcv{\texttt{VCVTT}\xspace}

\newcommand{\tool}{Palmed\xspace}

\newif\iffulldoc
\fulldoctrue        % fulldoc = true
\def\ndisjoint{\supset\hspace{-0.5em}\subset}
\def\ndisjoint{)\hspace{-0.44em}(}
\def\ndisjoint{>\hspace{-0.6em}<}

\begin{document}

%%%%%%%%%%%---SETME-----%%%%%%%%%%%%%
%%% Initial submission title:
% \title{\tool: Throughput Characterization for Any Architecture}
%%% ArXiV title:
\title{PALMED: Throughput Characterization for Superscalar Architectures -  Extended Version}
%%%%%%%%%%%%%%%%%%%%%%%%%%%%%%%%%

\maketitle

\begin{abstract}
  In a super-scalar architecture, the scheduler dynamically assigns micro-operations (\muops) to execution ports.
  The port mapping of an architecture describes how an instruction decomposes into \muops and lists for each \muop the set of ports it \emph{can} be mapped to.
  It is used by compilers and performance debugging tools to characterize the performance throughput of a sequence of instructions repeatedly executed as the core component of a loop.

  This paper introduces a dual \emph{equivalent} representation:
  The resource mapping of an architecture is an abstract model where, to be executed, an instruction \emph{must} use a set of abstract resources, themselves representing combinations of execution ports.
  For a given architecture, finding a port mapping is an important but difficult problem.
  Building a resource mapping is a more tractable problem and provides a simpler and equivalent  model.
  This paper describes \tool, a tool that automatically builds a resource mapping for pipelined, super-scalar, out-of-order CPU architectures.
  \tool does not require hardware performance counters, and relies solely on runtime measurements.

  We evaluate the pertinence of our dual representation for throughput modeling by extracting a representative set of basic-blocks from the compiled binaries of the SPEC CPU 2017 benchmarks. We compared the throughput predicted by existing machine models to that produced by \tool, and found comparable accuracy to state-of-the art tools, achieving sub-10 \% mean square error rate on this workload on Intel's Skylake microarchitecture.

\end{abstract}


\section{Introduction}
Performance modeling is a critical component for program optimizations, assisting compilers as well as developers in predicting the performance of code variations ahead of time.
Performance models can be obtained through different approaches that span from precise and complex simulation of a hardware description~\cite{Zesto,GEM5,PTLSim} to application level analytical formulations~\cite{Roofline,DBLP:journals/corr/HammerEHW17}.
A widely used approach for modeling the CPU of modern pipelined, super-scalar, out-of-order processors consists in decoupling the different sources of bottlenecks, such as the latency related ones (critical path, dependencies), the memory-related ones (cache behavior, bandwidth, prefetch, etc.), or the port throughput related ones (instruction execution units).
Decoupled models allow to pinpoint the source of a performance bottleneck, which is critical both for compiler optimization~\cite{GCC,LLVM}, kernel hand-tuning~\cite{BLIS,SPIRAL}, and performance debugging~\cite{CQA,IACA,OSACA,MIAMI,LLVM:MCA,uiCA}. In particular, the code selection step is based on ad-hoc instruction cost models, that \tool aims at automatically generating for new architectures.
Cycle-approximate simulators such as ZSim~\cite{ZSim} or MCsimA+~\cite{MCSim} can also take advantage of such an instruction characterization.
\emph{This paper focuses on modeling the port throughput, that is, estimating the performance of a dependency-free loop where all memory accesses are L1-hits.}

Such modeling is usually based on the so-called port mapping of a CPU, that is the list of execution ports each instruction can be mapped to.
This motivated several projects to extract information from available documentation~\cite{EXEgesis,OSACA}.
But the documentation on commercial CPUs, when available, is often vague or outright lacking information.
Intel's processor manual~\cite{Intel-manual}, for example, does not describe all the instructions implemented by Intel cores, and for those covered, it does not even provide the decomposition of individual instructions into micro operations (\muops), nor the execution ports that these \muops can use.

Another line of work that allows for a more exhaustive and precise instruction characterization is based on micro-benchmarks, such as those developed to characterize the memory hierarchy~\cite{MemH-jack}. While characterizing the latency of instructions is quite easy~\cite{AgnerFog,instlatx64.atw.hu,Granlund}, throughput is more challenging.
Indeed, on super-scalar processors, the throughput of a combination of instructions cannot be simply derived from the throughput of the individual instructions.
This is because instructions compete for CPU resources, such as functional units, or execution ports, which can prevent them from executing in parallel.
It is thus necessary to not only characterize the throughput of each individual instruction, but also to come up with a description of the available resources and the way they are shared.

In this work, we present a fully automated, architecture-agnostic approach, fully implemented in \tool, to automatically build a mapping between instructions and execution ports. It automatically builds a static performance model of the throughput of sets of instructions to be executed on a particular processor. While prior techniques targeting this problem have been presented, e.g. \cite{PMEvo,uops.info}, we make several key contributions in \tool to improve automation, coverage, scalability, accuracy and practicality:
\begin{itemize}[noitemsep,topsep=0em,leftmargin=10pt]
\item We introduce a dual equivalent representation of the port mapping problem, into a conjunctive abstract resource mapping problem, facilitating the creation of specific micro-benchmarks to saturate resources.
\item We present several new algorithms: to automatically generate versatile sets of saturating micro-benchmarks, for any instruction and resource; to build efficient Linear Programming optimization problems exploiting these micro-benchmark measurements; and to compute a complete resource mapping for \emph{all} benchmarkable instructions.
\item We present a complete, automated implementation in \tool, which we evaluate against numerous other approaches including IACA \cite{IACA}, LLVM-mca \cite{LLVM:MCA}, PMEvo \cite{PMEvo} and UOPS.info \cite{uops.info}.
\end{itemize}

This paper has the following structure.
Related work is first discussed in Sec.~\ref{sec:related}.
Sec.~\ref{sec:background} discusses the state-of-practice and presents our novel approach to automatically generate a valid port mapping.
Sec.~\ref{sec:contrib} presents formal definitions and the equivalence between our model and the three-level mapping currently in use.
Sec.~\ref{sec:algos} presents our architecture-agnostic approach to deduce the abstract mapping without the use of any performance counters besides elapsed CPU cycles.
Sec~\ref{sec:eval} extensively evaluates our approach against related work on two off-the-shelf CPU before concluding.

\section{Related Work}
\label{sec:related}
Intel has developed a static analyzer named IACA~\cite{IACA} which uses its internal mapping based on proprietary information.
However, the project is closed-source and has been deprecated since April 2019. Even though some latencies are given directly in the documentation~\cite{Intel-manual}, they are known to contain errors and approximations, in addition to being incomplete.

First attempts on x86 to measure the latency and throughput were led by Agner Fog~\cite{AgnerFog} and Granlund~\cite{Granlund} using hand-written microbenchmarks.
Fog also uses hardware performance counters and hand-crafted benchmarks to reverse-engineers port mappings for Intel, AMD and VIA CPUs.
Fog's mappings are considered by the community to be quite accurate.
For example, the machine model of the x86 back-end of the LLVM compiler framework~\cite{LLVM} is partially based on them~\cite{Agner:UsedByIntel}.
However, Fog and Granlund's approach is tedious and error-prone, since
modern CPU instruction sets have thousands of different intricate
instructions.
Abel and Reineke~\cite{uops.info, DBLP:journals/corr/abs-1911-03282} have tackled this problem by combining an automatic microbenchmark generator with an algorithm for port-mapping construction.
Their technique relies on hardware counters for the number of {\muop}s executed on each execution port, only available on recent Intel CPUs.
They recently started providing data on the newest generations of AMD CPUs, but
by lack of necessary hardware counters, only latency and throughput are
published.

OSACA~\cite{OSACA} is an open source alternative to IACA offering a similar static throughput and latency estimator.
It relies on automated benchmarks manually linked with publicly available documentation to infer the port mapping and the latencies of the instructions. The tool Kerncraft~\cite{DBLP:journals/corr/HammerEHW17} focuses on hot loop bodies from HPC applications while also modeling caches;
its mapping comes from automated benchmarks generated through Likwid~\cite{DBLP:conf/icppw/TreibigHW10} and hardware counters measurements.
CQA~\cite{DBLP:conf/hipc/RubialONJL14}, a static loop analyzer integrated into the MAQAO framework~\cite{DBLP:conf/hipc/DjoudiNJ08}, takes a similar path while also supporting OpenMP routines.
It combines dependency analysis, microbenchmarks, and a port mapping and previous manual results to offer various types of optimization advice to the user, such as vectorisation, or how to avoid port saturation.
Both Kerncraft and CQA use a hard-coded port mapping based on Fog's work and official Intel and AMD documentation.

Besides the classic port mappings, machine learning based approaches have also been used, \textit{eg.} in Ithemal~\cite{Ithemal}, to approximate the throughput of basic blocks with good accuracy.
However, the resulting model is completely opaque and cannot be analyzed or used for any other purpose than the prediction of basic block throughputs. For instance, Ithemal does not report on the influence of each instruction, which is critical for manual assembly optimization.

PMEvo~\cite{PMEvo} is a tool that, like \tool{}, automatically generates a set of benchmarks that it uses to build a port mapping.
It produces a disjunctive tripartite model with instructions, \muops{}, and ports, which is the key different with \tool{}.
It does not require hardware performance counter, and only relies on runtime measurements of its benchmarks.
The set of benchmarks used is determined semi-randomly using a genetic algorithm.
The benchmarks themselves are simpler than those used by \tool{} and contain at most two different types of instructions.
The main difference between PMEvo and \tool{} is that PMEvo uses internally a disjunctive bipartite resource model, instead of the conjunctive model used by \tool{}.
These models, while able to accurately predict the execution of pipelined instructions bottlenecked only on the execution ports, cannot represent other bottlenecks like the reorder buffer, or the non-pipelined instructions like division.
More importantly, PMEvo's approach is less scalable, as handling more instructions may quickly lead to an overwhelming number of microbenchmarks, while our approach is focused to generate specifically microbenchmarks that saturate resources. \tool{} can complete the full mapping, benchmarking included, in a few hours. Another key to this scalability is our incremental approach to handle complex instructions using a linear programming formulation to compute automatically, and optimally, the mapping.

\section{Motivation and Overview}
\label{sec:background}
% Note: OLD version of motivation is in "motivation_MICRO.tex"

\subsection{Background}

In this work, we consider a CPU as a processing device mainly described by the so-called ``port model''.
Here, instructions are first fetched from memory, then decomposed into one or more \emph{micro-operations}, also called \emph{\muops}. The CPU schedules these \muops on a free compatible execution port, before the final \emph{retirement} stage.
Even though some instructions such as \texttt{add \%rax, \%rax} translate into only a single \muop, the x86 instruction set also contains more complex instructions that translate into multiple \muops.
For example, the \texttt{wbinvd} (\textit{Write Back and Invalidate Cache}) instruction produces as many \muops as needed to flush every line of the cache, leading to thousands of \muops~\cite{uops.info}.

\emph{Execution ports} are controllers routing \muops to \emph{execution units} with one or more different functional capabilities: for example, on the Intel Skylake architecture, only the port~4 may store data;
and the store address must have previously been computed by an \emph{Address Generation Unit}, available on ports~2, 3 and~7.

The \emph{latency} of an instruction is the number of clock cycles elapsed between two dependent computations.
The latency of an instruction $I$ can be experimentally measured by creating a micro-benchmark that executes a long chain of instances of $I$, each depending on the previous one.

The \emph{throughput} of an instruction is the maximum number of instances of that instruction that can be executed in parallel in one cycle. On every recent x86 architecture, all units but the divider are fully pipelined, meaning that they can reach a maximum throughput of one \muop per cycle~--~even if their latency it greater than one cycle.
For an instruction $I$, the throughput of $I$ can be experimentally measured by creating a micro-benchmark that executes many non-dependent instances of $I$:
The combined throughput of a multiset\footnote{A multiset is a set that can contain multiple instances of an element. As with normal sets, the order of elements is not relevant} of instructions can be defined similarly.
For example, the throughput of $\{\addss^2,\bsr\}$, i.e. two instances of $\addss$ and one instance of $\bsr$, is equal to the number of instructions executed per cycle (IPC) by the micro-benchmark:

\begingroup
    \fontsize{7.8pt}{10pt}\selectfont
\begin{verbatim}
repeat:
  ADDSS %xmm1 %xmm1; ADDSS %xmm2 %xmm2; BSR %rax %rax;
  ADDSS %xmm3 %xmm3; ADDSS %xmm4 %xmm4; BSR %rbx %rbx;
  ADDSS %xmm5 %xmm5; ADDSS %xmm6 %xmm6; BSR %rcx %rcx;
  ...
\end{verbatim}
\endgroup

A \emph{resource-mapping} describes the resources used by each instruction in a way that can be used to derive the throughput for any multiset of instructions, without having to execute the corresponding micro-benchmark. Such information is crucial for manual assembly optimization to pinpoint the precise cause of slowdowns in highly optimized codes, and measure the relative usage of the peak performance of the machine.

In this work, we target the automatic construction of a resource mapping for a given CPU on which we can accurately measure elapsed cycles for a code fragment. Note that \tool only uses benchmarks that have no dependencies, that is, where all instructions can execute in parallel.
Consequently the order of instructions in the benchmark does not matter\footnote{We assume, like all related work we are aware of, that the CPU scheduler is able to optimally schedule these simple kernels.}.

\def\mappingsheight{5em}
\begin{figure*}[h!tb]
	\begin{subfigure}[b]{0.35\textwidth}\centering
		\includegraphics[height=\mappingsheight]{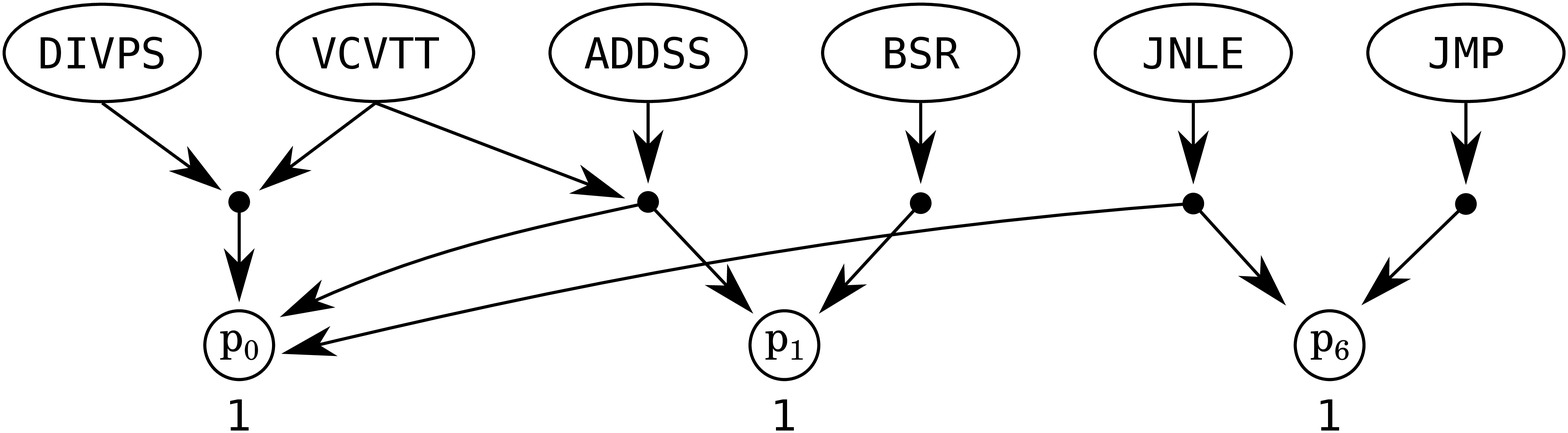}
		\caption{\label{fig:port-mapping}Port mapping (disjunctive form) and maximum throughput of each port.}
	\end{subfigure}
	\hfill
	\begin{subfigure}[b]{0.35\textwidth}\centering
		\includegraphics[height=\mappingsheight]{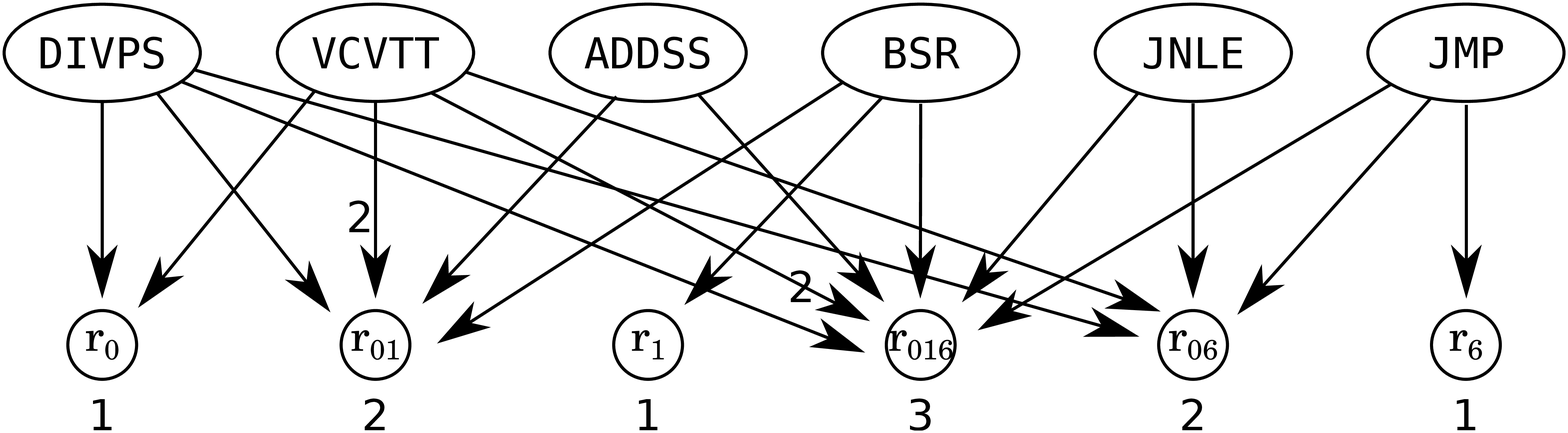}
		\caption{\label{fig:resource-mapping}Abstract resource mapping (conjunctive form) and maximum throughput of each resource.}
	\end{subfigure}\hfill
	\begin{subfigure}[b]{0.23\textwidth}\centering
		\includegraphics[height=\mappingsheight]{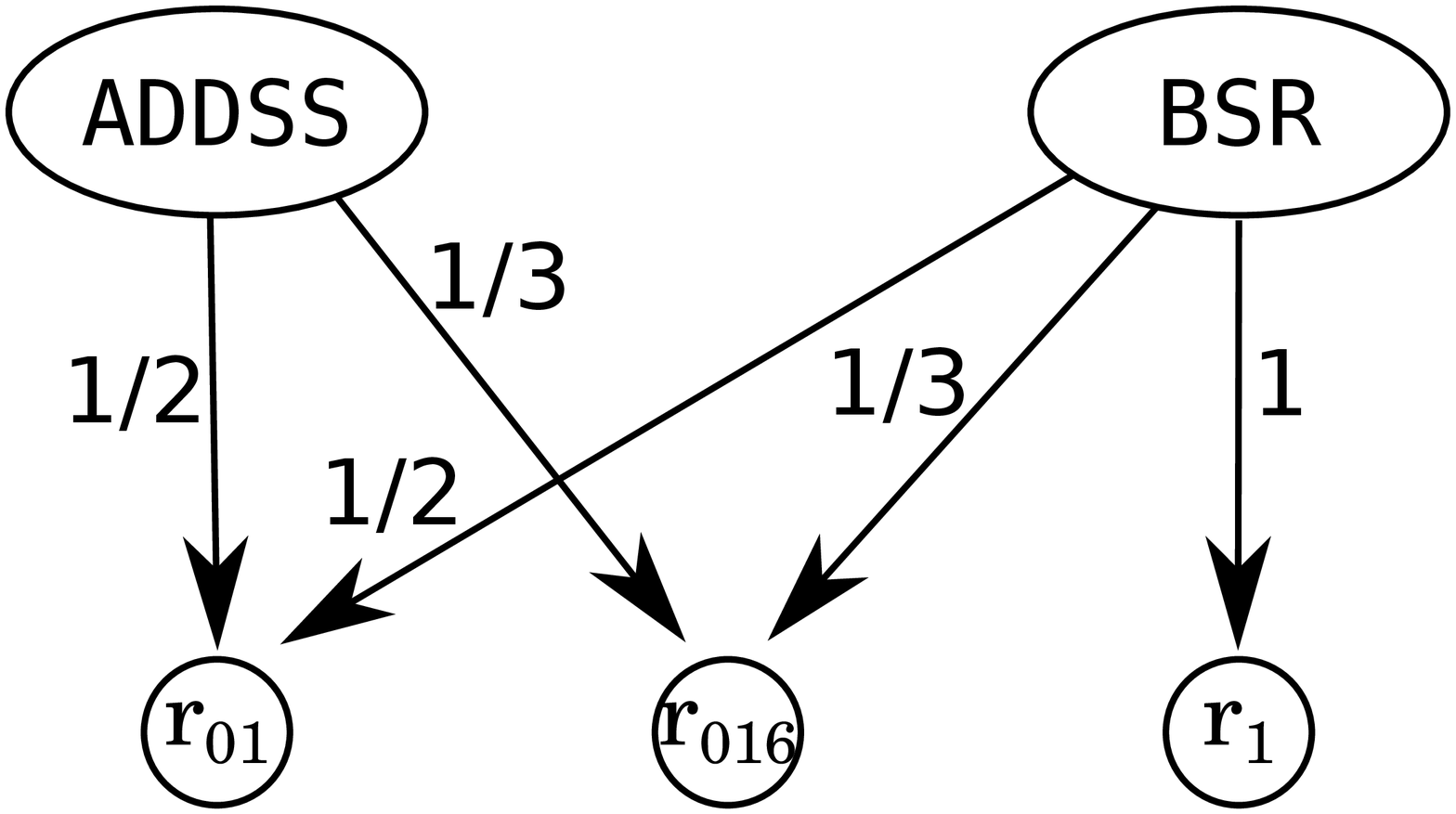}
		%\caption{\label{fig:resource-mapping-norm}Normalized abstract resource mapping, restricted.}
		\caption{\label{fig:resource-mapping-norm}Normalized conjunctive form for \addss and \bsr.}
	\end{subfigure}
	%\vspace{-.25cm}
	\caption{\label{fig:mapping}Mappings computed for a few SKL-SP instructions.\vspace{-.25cm}}
\end{figure*}

% ================
\subsection{Constructing a Resource Mapping}

To characterize the throughput of each individual instruction, a description of the available resources and the way they are shared is needed.
The most natural way to express this sharing is through a port mapping, a tripartite graph that describes how instructions decompose to \muops and assigns \muops to execution ports (see Fig.~\ref{fig:port-mapping}).
The goal of existing work has been to reverse engineer such a port mapping for different CPU architectures.

The first level of this mapping, from instructions to \muops, is conjunctive, i.e., a given instruction decomposes into one or more of each of the \muops it maps to.
The second level of this mapping, on the other hand, is disjunctive, \textit{i.e.} a \muop can choose to execute on any one of the ports it maps to.
Even with hardware counters that provide the number of \muops executed per cycle and the usage of each individual port, creating such a mapping is quite challenging and requires a lot of manual effort with ad hoc solutions to handle all the cases specific to each architecture~\cite{CQA,AgnerFog,Granlund,uops.info}.

Such approaches, while powerful and allowing a semi-automatic characterization of basic-block throughput, suffer from several limitations.
First, they assume that the architecture provides the required hardware counters.
Second, they only allow modeling the throughput bottlenecks associated with port usage, and neglect other resources, such as the front-end or reorder buffer.
Thus, it provides a performance model of an ideal architecture that does not necessarily fully match reality.

To overcome these limitations, we restrict ourselves to only using cycle measurements when building our performance model.
Not relying on specialized hardware performance counters may complicate the initial model construction, but in exchange our approach is able to model resources not covered by hardware counters with relative ease.
This also paves the way to significantly ease the development of modeling techniques for new CPU architectures.
One of the main challenges is to generate a set of micro-benchmarks that allows the detection of all the possible resource sharing.
Unfortunately, to be exhaustive, and in the absence of structural properties, this set is combinatorial: all possible mixes of instructions need to be evaluated.
A simple way to reduce the set of  micro-benchmarks required is to reduce the set of modeled instructions to those that are emitted by compilers~\cite{Ithemal,PMEvo}.
Another natural strategy followed by Ithemal~\cite{Ithemal} is to build micro-benchmarks from the ``most executed'' basic-blocks of some representative benchmarks.
A third strategy, used by PMEvo~\cite{PMEvo}, is to restrict micro-benchmarks to contain repetitions of only two different instructions.

Our solution is constructive and follows several successive steps that allow building a non-combinatorial number of micro-benchmarks that stresses the usage of each individual resource, thus characterizing the resource usage of \emph{all} instructions.

The second main challenge addressed by PMEvo is to build an interpretable model, that is, a resource-mapping that can be used by a compiler or a performance debugging tool, instead of a black-box only able to predict the throughput of a microkernel.
One issue with the standard port-mapping approach, as used in~\cite{uops.info,OSACA,LLVM:MCA}, is that computing the throughput of a set of instructions requires the resolution of a flow problem;
that is, given a set of micro-benchmarks, finding a mapping of \muop{}s to ports that best expresses the corresponding observed performances requires solving a multi-resolution linear optimization problem.
This linear problem also does not scale to larger sets of benchmarks, even when restricting the micro-benchmarks to only contain up to two different instructions.
PMEvo addressed this issue by using a evolutionary algorithm that approximates the result.

\begin{table}[h!tb]
\centering
\caption{\sc\label{tab:relatedworksummary}Summary of key features of \tool{} vs. related work}
{\small
\begin{tabular}{l c c c c}
\toprule
& no HW & no manual & \multirow{2}{*}{interpretable}  & \multirow{2}{*}{general}\\
& counters & expertise & & \\
\midrule
llvm-mca~\cite{LLVM:MCA} & \no & \no & \yes & \yes\\
Ithemal~\cite{Ithemal} & \yes & \no & \no & \no\\
IACA~\cite{IACA} & \na & \no & \yes & \yes\\
uop.info~\cite{uops.info} & \no & \no & \yes & \yes\\
PMEvo~\cite{PMEvo} & \yes & \yes & \yes & \no\\
{\bf \tool} & \yes & \yes & \yes & \yes\\
\bottomrule
\end{tabular}
}
\end{table}

% Note - CHANGE STARTING FROM HERE !!! (Charcutage de l'exemple de Fab - si acceptation, investir pages en priorité ici, pour réétablir l'ancienne version)
%\paragraph{Dual representation}
% ================
\subsection{Resource Mapping: Dual Representation}
\label{subsec:dual}

Our approach is based on a crucial observation: a dual representation exists for which computing the throughput is not a linear problem, but a simple formula instead.
While it takes several hours to solve the original disjunctive-port-mapping formulation, only a few minutes suffice for the corresponding conjunctive-resource-mapping formulation.

% Description of a conj/disj port mapping using Figure 1 (old Section 2.3)
% Def already in Sec 2.1
%A resource-mapping describes the resources used by each instruction in a way that can be used to derive the throughput for any multiset of instructions, without having to execute the corresponding micro-benchmark.

For the sake of illustration only (\tool finds in practice a mapping for all supported instructions), we consider the Skylake instructions restricted to those that only use ports~0, 1, or 6 (denoted as $p_0$, $p_1$, and $p_6$).
Fig.~\ref{fig:port-mapping} shows the port mapping for six such instructions.
In this example: the \muop of \bsr has a single port $p_1$ on which it can be issued;
as for instruction \addss, its \muop can be issued on either $p_0$ or $p_1$.
Hence, \bsr has a throughput of one, that is, only one instruction can be issued per cycle,
whereas \addss has a throughput of two: two different instances of \addss may be executed in parallel by $p_0$ and $p_1$.
The throughput of the multiset $K=\{\addss^2,\bsr\}$, more compactly denoted by $\addss^2\bsr$, is therefore determined by the combined throughput of resources $p_0$ and $p_1$.
Indeed, in a steady state mode, the execution can saturate both resources by repeating the pattern represented in Fig~\ref{fig:sched1}.
In this case, there clearly does not exist any better scheduling, and the corresponding execution time for $K$ is 3 cycles for every 6 instructions, that is, an Instruction Per Cycle (IPC) of 2.
Now, if we consider the set $\addss\;\bsr^2$, its throughput is limited by $p_1$.
Indeed, the optimal schedule in that case would repeat the pattern represented in Fig~\ref{fig:sched2}, which requires 2 cycles for 3 instructions, that is, an IPC of 1.5.
More generally, the maximum throughput of a multiset on a tripartite port-mapping can be solved by expressing the minimal scheduling problem as a flow problem.

\begin{figure}[h!tb]
\begin{subfigure}[b]{0.45\linewidth}
\centering
\footnotesize
\begin{tabular}{c|c}
\toprule
$p_0$ & $p_1$\\
\midrule
\addss & \bsr\\
\addss &  \bsr\\
\addss & \addss\\
\bottomrule
\end{tabular}
\caption{\label{fig:sched1}$\addss^2\;\bsr$}
\end{subfigure}
\begin{subfigure}[b]{0.45\linewidth}
\centering
\footnotesize
\begin{tabular}{c|c}
\toprule
$p_0$ & $p_1$\\
\midrule
\addss & \bsr\\
$\emptyset$ &  \bsr\\
\bottomrule
\end{tabular}
\caption{\label{fig:sched2}$\addss\;\bsr^2$}
\end{subfigure}

\caption{\label{fig:scheds_exemple}Disjunctive port assignment examples}
\vspace{-.2cm}
\end{figure}

The \emph{dual representation}, advocated in this paper, corresponds to a \emph{conjunctive} bipartite resource mapping as illustrated in Fig.~\ref{fig:resource-mapping}.
In this mapping, an instruction such as \addss{} which uses one out of two possible ports $p_0$ and $p_1$ will only use the abstract resource $r_{01}$ representing the combined load on both ports, and will use neither $r_0$ nor $r_1$.
In this model, the maximum throughput of $r_{01}$ is the sum of the throughput of $p_0$ and $p_1$, that is, 2 uses per cycle.
Instructions that may only be computed on $p_0$ will then use $r_0$ \emph{and} $r_{01}$, along with all other resources combining the use of $p_0$ with other ports such as $r_{06}$ and $r_{016}$.
Followingly, the average execution time of a microkernel is computed as the maximum load over all abstract resources, that is, their number of uses divided by their throughput (see Sec.~\ref{sec:contrib}).
One can prove (see~\cite{PMD-full}) the strict equivalence between the two representations \emph{without the need for any combinatorial explosion in the number of combined resources}. Because of this property, the trade-off offered by the conjunctive formulation (more resources for a simpler throughput computation) offers better overhaul solving complexity that former disjunctive-based approaches for real processors, hence the better scalability of \tool.
 %(e.g. 8 for Nehalem to 14 for Skylake-X):
%indeed, just as $r_{16}$ in Fig.~\ref{fig:resource-mapping} is subsumed by resources $r_1$ and $r_6$, many combined resources are redundant and not needed.
Indeed, in practice, some combined resources are not needed (e.g. $r_{16}$ in our example) as their usage is already perfectly described by the usage of individual resources (here, $r_1$ and $r_6$).

%This does not require solving a flow problem, but is just a simpler linear expression.

% Note: skipped some example here + Figure 3

%FAB: réduire cette partie qui est trop longue
A key contribution of this paper is to provide a less intricate two-level view, that can be constructed quicker than previous works.
Instead of representing the execution flow as the traditional three-level ``instructions decomposed as micro-operations (micro-ops) executed by ports'' model, we opt for a direct ``instructions use abstract resources'' model.
Whereas an instruction is transformed into several micro-ops which in turn \emph{may} be executed by different compute units;
our bipartite model \emph{strictly uses} every resource mapped to the instructions.
In other words, the \emph{or} in the mapping graph are replaced with \emph{and}, which greatly simplifies throughput estimation.
This representation may also represent other bottlenecks such as the instruction decoder or the reorder buffer as other abstract resources.
Note that this corresponds to the user view, where the micro-ops and their execution paths are kept hidden inside the processor.
An important contribution of this paper is to provide a constructive algorithm that provides a non-combinatorial set of representative micro-benchmarks that can be used to characterize all instructions of the architecture.

% ==== END CHANGE ====

\subsection{\tool{}: Flow of Work}
Fig.~\ref{fig:big_picture} overviews the major steps of \tool{}, which are extensively described in Sec.~\ref{sec:algos}.
Our algorithm follows an approach similar to the one developed by uops.info:
its principle is to first find a set of \emph{basic instructions} producing only one \muop and bound to one port.

This first step can be done on Intel CPUs by measuring the \muop per cycle on each port for each instruction through performance counters.

\begin{figure*}[h!tb]\
    \center
    \includegraphics[height=10em]{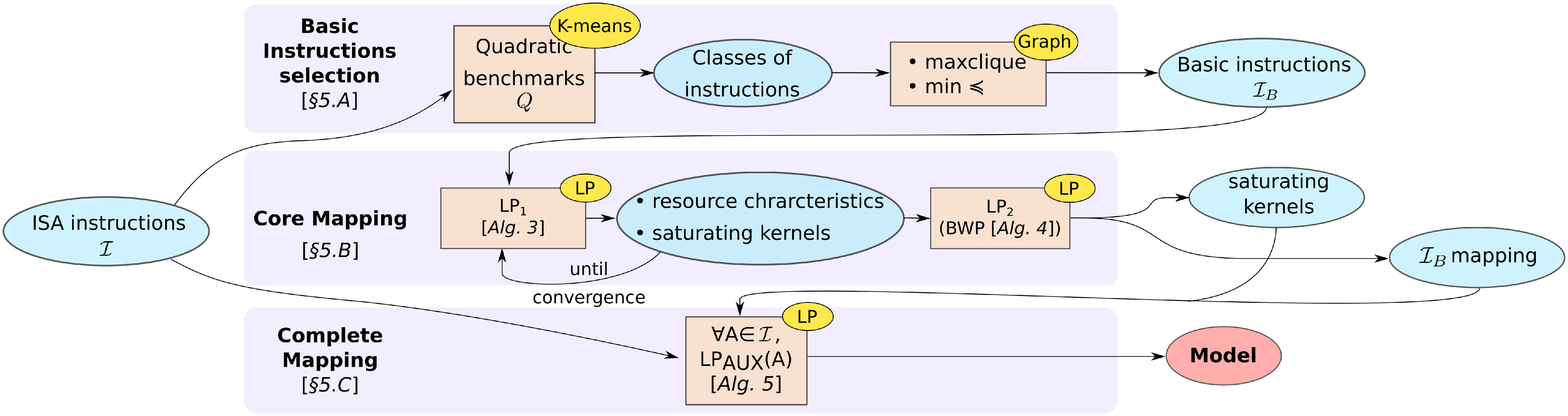}
    \caption{\label{fig:big_picture}High-level view of the algorithms of \tool{}}
\end{figure*}

Those basic instructions are then used to characterize the port mapping of any general instruction by artificially saturating one-by-one each individual port and measuring the effect on the usage of the other ports.
The challenge addressed by \tool is to find a mapping, even for architectures that do not have such hardware counters.

This translates in two major hardships:
firstly, in our case, there is no predefined resources; secondly, there even is no simple technique to find the number of \muops an instruction decomposes into.
As illustrated by Fig.~\ref{fig:big_picture} the algorithm of \tool is composed of three steps:
1.~Find basic instructions;
2.~Characterize a set of abstract resources (expressed as a \emph{core mapping}) and an associated set of saturating microkernels (a single instruction might not be enough to saturate a resource);
3.~Compute the resource usage of each other instruction with respect to the core mapping.

As an example, let us go back to our example: instructions using only $p_0$, $p_1$, or $p_6$. On Intel's Skylake microarchitecture, there exists 754 benchmarkable instructions using only these 3 ports.
Quadratic benchmarking --~that is, measuring the execution time of one benchmark per pair of instruction, leading to a quadratic number of measures (567762) --~ allows us to regroup those who have the same behavior together, leading to only 9 classes of instructions.
For each class, a single instruction is used as a representative.
Among those instructions, two heuristics (described in sec~\ref{sec:basic}) select the set of basic instructions, outputting \divps, \bsr, \jmp, \jnle, and \addss.

Fig.~\ref{fig:resource-mapping} shows the output of the \emph{Core Mapping} stage in Fig.~\ref{fig:big_picture}, in bold.
In practice, abstract resources are internally named $R_0$, \dots, $R_5$.
For convenience we renamed them to the hardware execution ports they correspond to: for example, the abstract resource $r_{01}$ corresponds to the combined use of port $p_0$ and $p_1$ for an optimal schedule.

The Core mapping also computes a set of saturating micro-benchmarks that individually saturate each of the individual abstract resource.
Here, each basic instruction will constitute by itself a saturating micro-benchmark: \divps will saturate $r_0$, \bsr will saturate $r_1$, \jmp will saturate $r_6$, \addss will saturate $r_{01}$, and \jnle will saturate $r_{06}$. Note that this is not the case in general:
we possibly need to combine several basic instructions together to saturate a resource. Here, the saturating micro-benchmark for resource $r_{016}$ is composed of two basic instructions: \addss and \jnle.
The last phase of our algorithm will, for each of the 742 remaining instructions, build a set of micro-benchmarks that combine the saturating kernels with the instruction, and compute its mapping.

\section{The Bipartite Resource Mapping}
\label{sec:contrib}

\iffulldoc
This section presents the theoretical foundation of \tool, that is, the equivalence between the disjunctive \muop{} mapping graph formerly used and a dual conjunctive formulation.
\else
This section provides a formal presentation of the dual conjunctive formulation used by \tool.
\fi

\begin{restatable}[Microkernel]{definition}{sec3_defmicrokernel}
\label{def:sec3_microkernel}
A microkernel $K$ is an \emph{infinite loop} made up of a finite multiset of instructions, $K=I_1^{\sigma_{K,1}}I_2^{\sigma_{K,2}}\cdots I_m^{\sigma_{K,m}}$ without dependencies between instructions, $\sigma_K$ representing the number of repetition of the instruction $K$ in the microbenchmark. The number of instructions executed during one loop iteration is $|K|=\sum_{i}\sigma_{K,i}$.
\end{restatable}
\vspace{-.2cm}

In a classical \emph{disjunctive} port mapping formalism, an instruction $i$ from a microkernel $K$ is assigned to a port (\emph{resource} $r$) that is compatible.
The execution time of $K$ is determined by the resource which is used the most by its instructions in a given such assignment, and depends on the assignment picked, as presented in Sec~\ref{sec:background}.
Instead, we consider a \emph{conjunctive} port mapping:
\vspace{-.2cm}
\begin{restatable}[Conjunctive port mapping]{definition}{sec3_defconjunctive}
\label{def:sec3_conjunctive}
A conjunctive port mapping is a bipartite \emph{weighted} graph $(I,\Resources,E,$ $\rho_{I,\Resources})$ where:
$I$ represents the set of instructions;
$\Resources$ represents the set of \emph{abstract resources}, that has a (normalized) throughput of 1;
$E \subset I\times \Resources$ expresses the required use of abstract resources for each instruction.
An instruction $i$ that uses a resource $r$ ($(i,r)\in E$) always uses the same proportion
(number of cycles) %, possibly lower/greater than 1)
$\rho_{i,r}\in \Q^+$.
If $i$ does not use $r$, then $\rho_{i,r}=0$.

Let $K=I_1^{\sigma_{K,I_1}}I_2^{\sigma_{K,I_2}}\cdots I_m^{\sigma_{K,I_m}}$ be a microkernel.
In a steady state execution of $K$, for each loop iteration, instruction $i$ must use a resource $r$ for $(\sigma_{K,i}\rho_{i,r})$ cycles.
The number of cycles required to execute one loop iteration is:
$$t(K)=\max_{r\in \Resources} \left(\sum_{i\in K} \sigma_{K,i}\cdot\rho_{i,r}\right)$$
\end{restatable}

One should observe that Def.~\ref{def:sec3_conjunctive} defines formally a \emph{normalized} version where throughputs of abstract resources are set to 1.
For the sake of clarity, the example in Sec.~\ref{sec:background} was considering non-normalized throughputs, that is, different than 1.
Going from non-normalized (as in Fig.~\ref{fig:resource-mapping}) to normalized form (as in Fig.~\ref{fig:resource-mapping-norm}) simply relies in dividing the incoming edges of a resource by the resource's throughput before setting its throughput to 1.
For example, in the non-normalized form \vcv{} uses 2 times $r_{01}$, which has a throughput of 2, leading to a normalized $\rho_{\vcv,r_{01}}$ of $1$.
Similarly, $\rho_{\addss,r_{016}}=1/3$.

\begin{restatable}[Throughput]{definition}{sec3_defthroughput}
\label{def:sec3_throughput}
The throughput $\ipc{K}$ of a microkernel $K$ is its instruction per cycle rate (IPC), defined as:
$$
\ipc{K} = \frac{|K|}{t(K)} = \frac{\sum_{i\in K} \sigma_{K,i}}{\max_{r\in \Resources} \sum_{i\in K} \sigma_{K,i}\cdot\rho_{i,r}}
$$
\end{restatable}

\paragraph*{\bf Example} If $K = \addss^2\;\bsr$, as in Fig~\ref{fig:sched1},
{\small
\begin{align*}
    t(K) &= \max\,_{r\in\{r_1, r_{01}, r_{016}\}}\left(2\times\rho_{\addss,r} + \rho_{\bsr,r}\right) \\
    &= \max\left(
    \resourcelabel{r_1}~2\times 0+1,
    \resourcelabel{r_{01}}~2\times\frac{1}{2}+\frac{1}{2},
    \resourcelabel{r_{016}}~2\times\frac{1}{3}+\frac{1}{3}
    \right)\\
    &= 1.5 \\
    \ipc{K} &= (2+1) / 1.5 = 2
\end{align*}}
On $K' = \addss\;\bsr^2$, as in Fig~\ref{fig:sched2}, the same computation gives $t(K') = 2$, the bottleneck being $r_1$; hence, $\ipc{K'} = 3/2$.

The mathematical definitions, the method to build a conjunctive port
mapping from a disjunctive one, and the abstract resource and the equivalence proof can be found
\iffulldoc
in Appendix~\ref{app:dual_repr_equiv}.
\else
in~\cite{PMD-full}.
\fi

\section{Computing Resource Mapping}
\label{sec:algos}
As depicted in Fig.~\ref{fig:big_picture}, our approach can be decomposed into three different steps.
Sec.~\ref{sec:basic} describes the selection of \emph{basic instructions}, a subset of instructions that map to as few resources as possible.
Sec.~\ref{sec:core_mapping} describes the computation of the \emph{core mapping} for these basic instructions.
The core mapping stays fixed for the rest of the algorithm.
Along with the core mapping, we also select a saturating micro-benchmark for each resource, called the \emph{saturating kernel}.
The saturating kernel is made up of basic instructions that do not place a heavy load on other resources.
Sec.~\ref{sec:complete_mapping} describes how \tool uses the saturating kernels and the core mapping to deduce, one by one, the resources used by the remaining instructions of the targeted architecture.

\subsection{Basic Instructions Selection}
\label{sec:basic}
The first step of our algorithm trims the instruction set to extract a minimal set of instructions for which the mapping will be computed.
As this (core) mapping will be reused later, we need enough instructions to detect all resources, but the more instructions we have, the longer the resolution of the linear problem this mapping will take.
We thus first apply two simple filters that reduce the number of basic instructions, as depicted in the first half of Algo.~\ref{alg:basic_inst}.

\paragraph*{\bf Low-IPC}
  If $\ipc{a}<1$ (measured with a microbenchmark repeating only $a$), then $a$ is not consider as a candidate for basic instructions. Assuming every physical resource to have a throughput of 1, such instructions use one resource more than once. However, these low-IPC instructions are still mapped at the very last step of \tool (see Sec.~\ref{sec:complete_mapping}).

  Then, we compute, for every remaining pair of instruction $(a,b)$, the throughput of the microkernel $a^{\ipc{a}}b^{\ipc{b}}$. This set of benchmarks is called \emph{quadratic benchmarks} (see Fig.~\ref{fig:big_picture}) as their number is quadratic with respect to the number of instructions. These quadratic benchmarks are later reused in each of the following heuristics.

\paragraph*{\bf Equivalent classes} If $\forall p,\ \ipc{a^{\ipc{a}}p^{\ipc{p}}} = \ipc{b^{\ipc{b}}p^{\ipc{p}}}$ then keep only $a$ or $b$.
  The second filter removes duplicates, that is, if two instructions behave similarly with regard to the evaluation used for our basic instruction selection, then one of them can be ignored.
  Obviously, on a real machine, despite all the crucial efforts to remove execution hazards, measured IPC never perfectly match and the correct criteria for selecting a representative instruction for duplicates should approximate the equality test $\forall p,\ \ipc{a^{\ipc{a}}p^{\ipc{p}}} \approx \ipc{b^{\ipc{b}}p^{\ipc{p}}}$.
  The construction of equivalence classes and associated representative instruction in this context uses hierarchical clustering~\cite{hclust}.

% ======== ALGO 1
\begin{algorithm}[ht!]
\small

  \Fn{Select\_basic\_insts($\Instructions, n$)}{
    $\Instructions_F := \Instructions$\;
    \tcp{Remove low-IPC; compute eq. classes}
    \ForEach{$a\in \Instructions_F$}{
      \lIf{$\ipc{a} \leq 1 - \epsilon$}{
        $\Instructions_F:=\Instructions_F-\left\{a\right\}$
      }
      \If{$\exists b\in \Instructions_F,\ \forall p\in\Instructions,\ \ipc{a^{\ipc{a}}p^{\ipc{p}}} = \ipc{b^{\ipc{b}}p^{\ipc{p}}}$}{
        $\Instructions_F:=\Instructions_F-\left\{a\right\}$
      }
    }
    \tcp{Select very basic instructions}
    \ForEach{$a\in \Instructions_F$}{
      $\textit{Dj}[a] := \left\{b \in \Instructions_F, a^{\ipc{a}}b^{\ipc{b}} = \ipc{a} + \ipc{b}\right\}$
    }
    \textbf{let}
    $a <_\textit{VB} b \Leftrightarrow $ \\ ~\hfill$\left(\left|\textit{Dj}[a]\right|>\left|\textit{Dj}[b]\right|\right) \vee \left(\left|\textit{Dj}[a]\right|=\left|\textit{Dj}[b]\right| \wedge \ipc{a}>\ipc{b}\right)$\;
    $\Instructions_\textit{VB}:=\emptyset$\;
    \For{$a \in \Instructions_F$ in $<_\textit{VB}$ order}{
      \lIf{$\Instructions_\textit{VB}\subset \textit{Dj}[a]$}{
        $\Instructions_\textit{VB} := \Instructions_\textit{VB} \cup \{ a \}$
      }
      \lIf{$\left|\Instructions_{VB}\right|=n$}{\textbf{return} $\Instructions_B:=\Instructions_{VB}$}
    }
    \tcp{Select most greedier instructions}
    $\Instructions_\textit{MF}:=\emptyset$\;
    \For{$a \in \Instructions_F$ in $\preccurlyeq_\textit{greedier}$ order}{
      $\Instructions_\textit{MF}:=\Instructions_\textit{MF} \cup \{ a \}$\;
      \lIf{$\left|\Instructions_{VB}\cup\Instructions_\textit{MF}\right|=n$}{\textbf{return} $\Instructions_B:=\Instructions_\textit{VB}\cup\Instructions_\textit{MF}$}
    }
    \Return{$\Instructions_B:=\Instructions_\textit{VB}\cup\Instructions_\textit{MF}$}\;
  }
\caption{\label{alg:basic_inst}Set of basic instructions $\Instructions_B$}
\end{algorithm}

% == ALGO 2
\begin{algorithm}[ht!]
\vspace{2mm}
\small
    \Fn{Core\_mapping($\Instructions_B$)} {
        \tcp{Characterize resources}
        $\Kernels := \bigcup_{(a,b)\in \Instructions_B^2,\ a\ne b}\left\{a,\ a^{\ipc{a}}b^{\ipc{b}},\ a^Mb\right\}$\;
         \DoUntil{$\Kernels_\textit{new}=\emptyset$} {
          $\mathcal{G} := \text{Shape\_Mapping}(\Kernels,\Instructions_{VB}, \Instructions_{MF})$) \tcp*{$LP_1$}
          $\Kernels_\textit{new}:=\bigcup_{r\in\Resources} \left\{\Pi_{i\in\Instructions_B,\ \rho_{i,b}\ge \epsilon} i^{\ipc{i}}\right\}-\Kernels$\;
          $\Kernels:=\Kernels\cup\Kernels_\textit{new}$\;
        }
        $\mathcal{G} := \text{Mapping}(\Kernels, \mathcal{G}$) \tcp*{$LP_2$}
        \tcp{Find saturating kernels}
        \ForEach{$r\in \mathcal{R}$}{
          $\textit{sat}[r] := K \in \Kernels \text{ s.t. } \rho_{K,r}=1 \text{ minimizing } \textit{cons}(K)$\;
          \For{$i \in \Instructions_B$ s.t. $i \notin \textit{sat}[r]$} {
            $\Kernels := \Kernels \cup \{K_\textit{sat}(i,r)\}$\;
          }
        }
        \Return $\Kernels, \textit{sat}, \mathcal{G}$\;
    }
\caption{\label{alg:core-mapping}Core mapping and saturating kernels}
\end{algorithm}

% == ALGO 3
\begin{algorithm}[ht!]
\small
  \Fn{$Shape\_mapping(\Kernels,\Instructions_{VB}, \Instructions_{MF})$} {
    	\Solve{}{
    	$\forall (i,r) \in \Instructions \times \Resources,\; \rho_{i,r} \in \{0,1\}$\;
    	$\forall i\in \Instructions_{VB}, \min_{r\in \Resources} 1-\rho_{i,r} + \sum_{j \in \Instructions_{VB}\setminus\{i\}} \rho_{j,r} = 0$\;
    	$\forall i \in \Instructions_{MF}, \max_{r\in \Resources} \rho_{i,r} + \sum_{j\ndisjoint i} \rho_{j, r} = 1+|\{j \ndisjoint i\} |$\;
    	\ForEach{$k \in \Kernels$ s.t. $\{ i^\alpha\in k \text{ s.t. } \text{cycles}(i^\alpha) = \text{cycles}(k)\} = \emptyset$ }{
    		$\max_{r\in\Resources} \sum_{i\in k} \rho_{i,r} \geq |\{i \in k\}|$\;
    	}
    	
    	\ForEach{$k \in \Kernels$ s.t. $\{ i^\alpha\in k \text{ s.t. } \text{cycles}(i^\alpha) = \text{cycles}(k)\} \neq \emptyset$ }{
    		$\forall i^\alpha\in k \text{ s.t. } \text{cycles}(i^\alpha) = \text{cycles}(k)$\\\hfill
    		$\min_{r\in \Resources} 1-\rho_{i, r} + \sum_{j\in k, j\neq i}\rho_{j, r} = 0$\;
    	}
    	Minimize $\sum_{i\in \Instructions_B} \max_{r\in \Resources} \rho_{i,r}$\;
    }
    \Return $(\Instructions, \Resources, \{\rho_{i,r}\})$\;
  }
\caption{\label{alg:kernel_lp1}$\texttt{LP}_1$: Shape of core mapping}
\end{algorithm}

% ALGO 4
\begin{algorithm}[ht!]
\small
  \Fn{Mapping($\Kernels, \mathcal{G}$)}{
  \Solve{Bipartite Weight Problem}{
  $\Instructions := $ instructions($\Kernels$)\;
  $(\rho_{i,r})_{\Instructions, \Resources} := $ edges($\mathcal{G}$)\;
    $\forall (i,r)\in\Instructions\times \Resources,\; 0\le \rho_{i,r} \in [0,1]$\;
    $\forall (K,r)\in \Kernels\times \Resources,$\\
        ~\hfill $\rho_{K,r}=\left(\sum_{i\in\Instructions} \sigma_{K,i}\rho_{i,r}\right)\times \ipc{K}/\left(\sum_{i\in\Instructions}\sigma_{K,i}\right)$\;
    $\forall (K,r)\in \Kernels\times \Resources,\; \rho_{K,r}\le 1$\;
    $\forall K\in \Kernels,\;S_K = \max_{r\in\Resources} \rho_{K,r}$\;
    Minimize $\sum_{K \in \Kernels} (1-S_K)$\;
}
   \Return $(\Instructions, \Resources, \{\rho_{i,r}\})$\;
  }
\caption{\label{alg:kernel_lp2}$\texttt{LP}_2$: Bipartite Weight Problem (BWP), used in $\texttt{LP}_2$ and $\texttt{LP}_\texttt{AUX}$}
\end{algorithm}

% ALGO 5
\begin{algorithm}[ht!]
\small
    $\Instructions_B:=\textrm{select\_basic\_insts}(\Instructions,n)$\;
    $\Kernels, \textit{sat}, \mathcal{G}:= \textrm{Core\_mapping}(\Instructions_B)$\;
    \ForEach{$\textit{inst} \in \Instructions$}{
      $\Kernels:=\bigcup_{r\in\Resources} K_\textit{sat}(\textit{inst},r)$\;
      $\Instructions := \Instructions_B \cup \{\textit{inst}\}$\;
        \Solve{Find a solution to the following problem}{
          $\forall r\in\Resources,\; 0\le \rho_{\textit{inst},r}$\;
          $\forall (K,r)\in \Kernels\times \Resources,\; \rho_{k,r}=\left(\sum_{i\in\Instructions} \sigma_{K,i}\rho_{i,r}\right)\times \ipc{k}/\left(\sum_{i\in\Instructions}\sigma_{K,i}\right)$\;
          $\forall (K,r)\in \Kernels\times \Resources,\; \rho_{K,r}\le 1$\;
          $\forall K\in \Kernels,\;S_K = \max_{r\in\Resources} \rho_{K,r}$\;
    	      Minimize $\sum_{K \in \Kernels} (1-S_K)$\;
        }
    }
\caption{\label{alg:final_mapping}$\texttt{LP}_\texttt{AUX}$: Complete resource mapping}
\end{algorithm}

Once low IPC instruction duplicates have been removed, the selection relies on two criteria (cf Algo.~\ref{alg:basic_inst}):
%, as described in Algo.~\ref{alg:basic_inst}:

%\paragraph*{\bf Very basic instructions}
\begin{itemize}
\item \textbf{Very basic instructions:}
Instructions $a$ and $b$ are considered \emph{disjoint} if $\ipc{a^{\ipc{a}}b^{\ipc{b}}} = \ipc{a} + \ipc{b}$.
  The set of very basic instructions is defined as a maximal clique of disjoint instructions.
%The idea is to capture instructions that maps to a single resource.
This captures instructions that maps to a single resource.
Indeed, two instructions that do not share any resource
will have their IPC additive, thus belonging to the maximum clique of our graph.

%\paragraph*{\bf Most greedier instructions}
\item \textbf{Most greedier instructions:}
  Instruction $a$ is considered more greedier than $b$ ($a \preccurlyeq_{greedier} b$) if $\forall p, \ipc{a^{\ipc{a}}p^{\ipc{p}}} \geq \ipc{b^{\ipc{b}}p^{\ipc{p}}}$.
  This relation defines a pre-order, and we select the $n$ most greedier instructions.% (the bigger is $n$ the more complete is the core mapping but also the more complex is the linear program).
\end{itemize}
%We want here to collect instructions that use the lowest number of resources.

\subsection{Core Mapping}
\label{sec:core_mapping}

The core mapping phase as described in Alg.~\ref{alg:core-mapping} is decomposed in two steps.
The objective of the first step is to build the shape of the resource mapping from the basic instructions, containing all visible resources and possible edges.
In the second step, \tool computes the values of the edges and outputs a \emph{saturating benchmark} for every detected resource. Similarly to Abel and Reinecke's work~\cite{uops.info}, these benchmarks are reused as an indicator of the usage of their resource in the complete mapping phase (Sec.~\ref{sec:complete_mapping}).

Both of these steps use linear programming to build step-by-step a mapping that reflects accurately the measured IPC of a set of microkernels $\Kernels$, and are detailed in the following two paragraphs.

\paragraph*{\bf Characterize resources ($\texttt{LP}_1$)}
The goal of the first step is to find the shape of the resource mapping, that is, the number of resources needed and the possible edges from core instructions to resources. For this, \tool{} solves the following Integer Linear Programming (ILP) problem, formalized in Alg.~\ref{alg:kernel_lp1}, repeated until no new benchmark is added:
\begin{itemize}[leftmargin=*, label={}, itemsep=0.05cm]
\item \textbf{Objective function:} Minimize the number of resources.
\item \textbf{Constraints:}
From the following seed of microkernels:% defined as follow:
\begin{enumerate}[noitemsep, topsep=0cm]
\item $a \in \Instructions$ alone;
\item $a^{\ipc{a}}b^{\ipc{b}}$, as this benchmark has and IPC of $\ipc{a} + \ipc{b}$ if $a$ and $b$ are independents or if their common resources are not dominantly used.
% as this benchmark has the following property:
%If $a$ and $b$ are independent, that is the set of resources used by $a$ and $b$ are disjoint, or have a cumulated usage that does not exceed $\frac{1}{\ipc{a} + \ipc{b}}$, then $\ipc{a^{\ipc{a}}b^{\ipc{b}}} = \ipc{a} + \ipc{b}$;
\item $a^{M} b$ (with $M = 4$ in practice --~see~\cite{PMD-full} for detailed justification) to avoid the convergence of the solver to a simpler solution with fewer resources.
\end{enumerate}

We derive the following constraints (in the order of Alg.~\ref{alg:kernel_lp1}):
\begin{itemize}[noitemsep, topsep=0cm]
\item Each very basic instruction as defined in Sec.~\ref{sec:basic} is linked to at least one resource unused by other very basic instructions (line 4).
\item For each greedier instruction $i$ as defined in Sec.~\ref{sec:basic}, there exists at least one resource common to $i$ and to all other instructions $a$ for which $\ipc{i^{\ipc{i}}a^{\ipc{a}}} \neq \ipc{i} + \ipc{a}$ (line 5). This relation corresponds to the negation of the \emph{disjoint} relation defined in Sec.~\ref{sec:basic}, that we note $\ndisjoint$.
\item For all other microkernels: 1) every instruction identified as saturating (that is, instructions for which the execution time of the microkernel is equal to its execution time) maps to at least a resource unused by other instructions of the microkernel (line 7); 2) if no saturating instruction is found, then there exists a resource shared by every instruction of the benchmark (line 10).
\end{itemize}
\end{itemize}

The enrichment is done as follows:
for each resource found, we add a benchmark composed of every instruction using it with a multiplicity of their IPC, splitting it in case of undesired merges.
Once convergence has been reached, we expect most of existing resources and edges to be discovered: \tool{} passes to the ($\texttt{LP}_2$) step to compute the value of the edges.

\paragraph*{\bf Bipartite Weight Problem (BWP) and Core Mapping ($\texttt{LP}_2$)}
The BWP is formalized in Alg.~\ref{alg:kernel_lp2}, and aims at finding the correct values of the edges found during the $\texttt{LP}_1$.
Using the notations from Def.~\ref{def:sec3_conjunctive}: $\rho_{i,r}\in \Q^+$ expresses the proportional usage of the resource $r$ by instruction $i$, and $\ipc{K}$ the average number of instructions executed each cycle when $K$ is executed by the CPU.
The proportion of a resource $r$ that is used is thus  $\rho_{K,r}=\overline{K}\cdot \left(\sum_{i\in\Instructions} \sigma_{K,i}\rho_{i,r}\right)/\left(\sum_{i\in\Instructions}\sigma_{K,i}\right)$, bounded by its throughput ($\rho_{K,r}\le \rho_r=1$).
One of the resources must be the limiting factor, that is, $\exists r,\ \rho_{K,r}=1$. However, we authorise sub-saturation of the resources, acknowledging our model does not predict accurately every microkernel, and we note $S_K = \max_r \rho_{K,r} \leq 1$. We also restrict the possible edges to the ones output by the $\texttt{LP}_1$.
These constraints form our linear problem minimizing the sum of predictions error, that is $\sum_{k\in \Kernels} (1 - S_k)$.

Once the mapping has been computed, for every resource $r$, a saturating kernel $\textit{sat}[r]$ is chosen among all saturating microbenchmarks of the $\texttt{LP}_2$ ($K$ s.t. $\rho_{K,r}=1$, at least one necessarily exists by construction) as the one that has minimum consumption:
$$\textit{cons}(K)=\sum_{i\in \Instructions,\ r\in \Resources} \rho_{i,r}$$

\subsection{Complete Mapping ($\texttt{LP}_\texttt{AUX}$)}
\label{sec:complete_mapping}
In the last step, corresponding to Algo.~\ref{alg:final_mapping}, an optimization problem is solved for each remaining instruction.
The formulation of the new optimization problem is very similar to the BWP, except that the resources and the edges of the core mapping computed previously are frozen.
The presence or absence of an edge from the to-be-mapped instruction $i$ to a resource $r$ is constrained by using $K_\textit{sat}(i,r) = i^{\ipc{i}} \textit{sat}[r]^{L*\ipc{\textit{sat}[r]}}$ in the set of microbenchmarks, with $L=4$ in practise. The idea is to force the saturation of $r$ by charging it with $\textit{sat}[r]$, hence expressing the usage of $r$ by $i$.

\section{Evaluation}
\label{sec:eval}
Our evaluation section compares throughput accuracy on assembly microkernels extracted from two benchmarks suites: the SPECrate version of SPECint2017~\cite{SPECCPU2017} and Polybench~\cite{polybench}.

We compare \tool{} against the native execution, along with the predictions of four existing tools: IACA~\cite{IACA}, PMEvo~\cite{PMEvo}, llvm-mca~\cite{LLVM:MCA} and the port mapping deduced from uops.info's work~\cite{uops.info}.

Our evaluation is performed on two architectures: the SKL-SP is an Intel Xeon Silver 4114 CPU at 2.20 GHz, using Debian, Linux kernel 4.19 and PAPI 6.0.0.1 to collect the execution time in cycles for each microbenchmarks, restraining to non-AVX-512 instructions. The ZEN is an AMD EPYC 7401P CPU at 2 GHz with a similar software setup.
For each of these two architectures, the number of generated microbenchmarks, resources found and mapped instructions are gathered in Table~\ref{fig:test_environment}.

\subsection{Calibration of the Model}
\label{sec:model-calibration}
The port mapping is computed using the algorithm presented in Section~\ref{sec:algos} using a list of x86 instructions extracted from Intel's XED~\cite{Intel-XED}. We discard instructions which cannot be instrumented in practice, such as instruction modifying the control flow (as our microbenchmark generator cannot handle non-trivial control flow in the instrumented instructions), privileged instructions, along with instructions whose IPC is lower than 0.05, as they do not present any interest for performance prediction of throughput-limited microkernels. While benchmarking memory instructions, we ensure that every access hits the L1 cache to avoid cache-related bottlenecks, which are out of \tool's scope. Due to the complexity of the x86 instruction set, we separate the SSE and AVX instructions from the ``base ISA'': we apply separately the heuristics of Sec.~\ref{alg:basic_inst} before gathering all selected instructions in a single combined \emph{basic instructions}' set as described in Fig.~\ref{fig:big_picture}.

We also forbid benchmarks combining different extensions (e.g. SSE+AVX). Indeed, combinations of several vector extensions of different width are known to cause extra latency, that is, a sort of dependency from one instruction to the other (two consecutive SSE instructions would not be penalized, whereas one SSE and one AVX will). This violates our assumption that the relative order of instruction does not matter, and in practice we observed a significant degradation of the mapping without this mitigation.

Because of variations in the real-world measurements, we constrain the error rate to $0.05$ for the micro-benchmark coefficient, meaning that the number of repetitions of an instruction inside its microkernel differs by at most 5\% from what the algorithm requires.
For example, a benchmark $a^{\ipc{a}}b^{\ipc{b}}$ with $\ipc{a}~=~0.06$ and $\ipc{b} = 1$ will be rounded to $a^1b^{20}$. Note that in the BWP defined in Algorithm~\ref{alg:kernel_lp2}, we use the rounded coefficients and not the ideal ones.
The IPC is also rounded accordingly. Note that our microbenchmark generator is pre-constrained with these limitations; therefore we did not evaluate \tool with another measurement back-end~-- although we expect similar results as we ensured to have reproducible execution times.

\begin{figure*}
    \begin{subfigure}[c]{\linewidth}
        \centering
        \begin{figleftlabel}{}
            \begin{minipage}[c]{0.03\linewidth}~\end{minipage}\figspaceleft{}
            \begin{minipage}[c]{0.13\linewidth}
                \figcollegend{\tool{}}
            \end{minipage}\figspacemid{}
            \begin{minipage}[c]{0.13\linewidth}
                \figcollegend{uops.info}
            \end{minipage}\figspacemid{}
            \begin{minipage}[c]{0.13\linewidth}
                \figcollegend{PMEvo}
            \end{minipage}\figspacemid{}
            \begin{minipage}[c]{0.13\linewidth}
                \figcollegend{IACA}
            \end{minipage}\figspacemid{}
            \begin{minipage}[c]{0.13\linewidth}
                \figcollegend{llvm-mca}
            \end{minipage}\figspaceright{}
            \begin{minipage}[c]{0.035\linewidth}~\end{minipage}
        \end{figleftlabel}

        \begin{figleftlabel}{SKL-SP}
            \figfiverow{SPEC 2017}{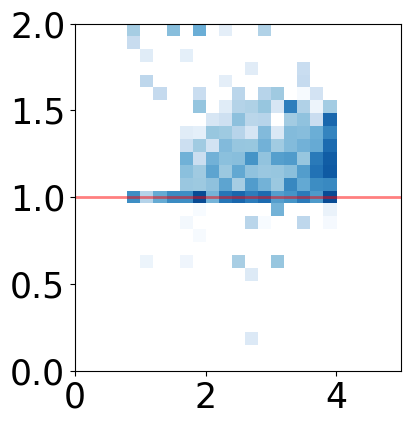}

            \figspacerow

            \figfiverow{Polybench}{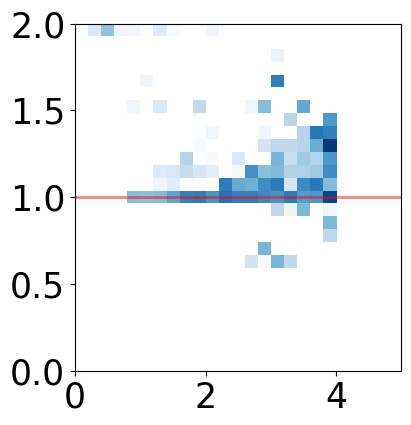}
        \end{figleftlabel}

        \figspacerow{}

        \begin{figleftlabel}{ZEN1}
            \figthreerow{SPEC 2017}{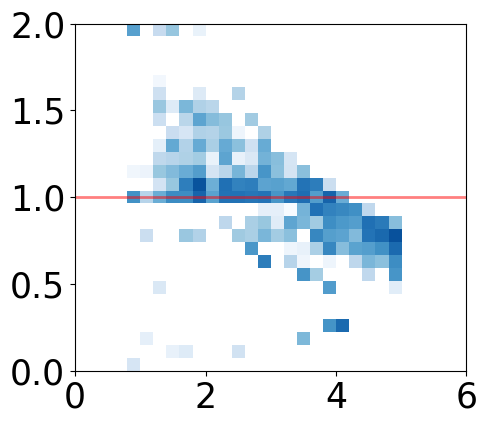}

            \figspacerow

            \figthreerowlegend{Polybench}{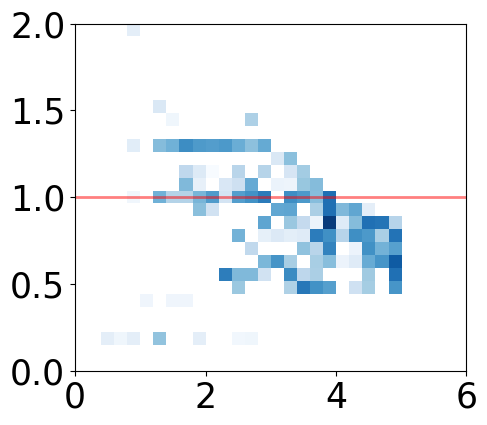}
        \end{figleftlabel}

        \caption{IPC prediction profile heatmaps~--~predictions closer to the
            red line are more accurate. Predicted IPC ratio (Y) against native
            IPC (X)}
        \label{fig:realcpu_heatmap}
    \end{subfigure}

    \vspace{1em}

    \begin{subtable}[c]{\linewidth}
    {
        \centering
        \caption{Translation block coverage (Cov.), root-mean-square error on
        IPC predictions (Err.) and Kendall's tau correlation coefficient
    ($\tau_K$) compared to native execution}
		\footnotesize

%%%%%%%%%%%%%%%%%%%%%%%%%%%%%%%%%%%%%%%%
% Table generated automatically using palmed's `tools/gen_results_table.py`
%   on 2021-08-27 17:22:33.207721
%   by tobast@patate
\label{tab:realcpu_results}\begin{tabular}{c c | C E K | C E K | C E K | C E K | C E K}
    \toprule
    & & \multicolumn{3}{c}{PMD} & \multicolumn{3}{c}{uops.info} & \multicolumn{3}{c}{PMEvo} & \multicolumn{3}{c}{IACA} & \multicolumn{3}{c}{llvm-mca}\\
    & & \tabcoverage & \taberror & \tabkendall & \tabcoverage & \taberror & \tabkendall & \tabcoverage & \taberror & \tabkendall & \tabcoverage & \taberror & \tabkendall & \tabcoverage & \taberror & \tabkendall\\
    & Unit & \tabcoverageunit & \taberrorunit & \tabkendallunit & \tabcoverageunit & \taberrorunit & \tabkendallunit & \tabcoverageunit & \taberrorunit & \tabkendallunit & \tabcoverageunit & \taberrorunit & \tabkendallunit & \tabcoverageunit & \taberrorunit & \tabkendallunit\\
    \midrule
    % Generated from /home/tobast/src/palmed/results/data/scw-skx.2021-08-26.benchs.spec17
    \multirow{2}{*}{\textbf{SKL-SP}} & \textbf{SPEC2017} & \na & 7.8 & 0.90 & 99.9 & 40.3 & 0.71 & 71.3 & 28.1 & 0.47 & 100.0 & 8.7 & 0.80 & 96.8 & 20.1 & 0.73\\
    % Generated from /home/tobast/src/palmed/results/data/scw-skx.2021-08-26.benchs.polybench
     & \textbf{Polybench} & \na & 24.4 & 0.78 & 100.0 & 68.1 & 0.29 & 66.8 & 46.7 & 0.14 & 100.0 & 15.1 & 0.67 & 99.5 & 15.3 & 0.65\\
    % Generated from /home/tobast/src/palmed/results/data/scw-zen.2021-08-24.benchs.spec17
    \multirow{2}{*}{\textbf{ZEN1}} & \textbf{SPEC2017} & \na & 29.9 & 0.68 & \na & \na & \na & 71.3 & 36.5 & 0.43 & \na & \na & \na & 96.8 & 33.4 & 0.75\\
    % Generated from /home/tobast/src/palmed/results/data/scw-zen.2021-08-24.benchs.polybench
     & \textbf{Polybench} & \na & 32.6 & 0.46 & \na & \na & \na & 66.8 & 38.5 & 0.11 & \na & \na & \na & 99.5 & 28.6 & 0.40\\
    \bottomrule
\end{tabular}
% End automatically generated table
%%%%%%%%%%%%%%%%%%%%%%%%%%%%%%%%%%%%%%%%

        }
    \end{subtable}

    \caption{Accuracy of IPC predictions compared to native execution of \tool versus uops.info, PMEvo, IACA and llvm-mca on SPEC CPU2017 and PolyBench/C 4.2}
        \label{fig:realcpu_results}

\end{figure*}

\subsection{Throughput Estimations}

To evaluate \tool, the same microkernel is run:
(1)~natively on each CPU, with the IPC measured with \texttt{CPU\_CLK\_UNHALTED};
(2)~using our mapping with abstract resources corresponding to the actual machine, as described in Section~\ref{sec:model-calibration};
(3)~using Abel's work (uops.info)~\cite{uops.info}, by running a conjunctive mapping with exact compatibility and approximating the execution time by the port with the highest usage;
(4)~using PMEvo~\cite{PMEvo}, ignoring any instruction not supported by its provided mapping;
(5)~using IACA~\cite{IACA}, by inserting assembly markers around the kernel and running the tool;
(6)~using llvm-mca~\cite{LLVM:MCA}, by inserting markers in the assembly code generated by our back-end and running the tool with this assembly.

%\begin{enumerate}[noitemsep,topsep=0pt]
%    \item{} natively on each CPU, with the IPC measured with \texttt{CPU\_CLK\_UNHALTED};
%    \item{} using our mapping with abstract resources corresponding to the actual machine, as described in Section~\ref{sec:model-calibration}.
%    \item{} using Abel's work (uops.info)~\cite{uops.info}, by running the conjunctive mapping with exact compatibility found in Section~\ref{sec:against_uops} and approximating the execution time by the abstract resource with the highest usage;
%    \item{} using PMEvo~\cite{PMEvo}, ignoring any instruction not supported by its provided mapping;
%    \item{} using IACA~\cite{IACA}, by inserting assembly markers around the kernel and running the tool;
%    \item{} using llvm-mca~\cite{LLVM:MCA}, by inserting markers in the assembly code generated by our back-end and running the tool with this assembly.
%\end{enumerate}

Unlike PMEvo and llvm-mca, UOPS and IACA do not support the ZEN1 architecture; hence the absence of data.

The microkernels are extracted from two well-known benchmark suites: SPECInt2017~\cite{SPECCPU2017} and Polybench~\cite{polybench}. For Polybench, we used QEMU~\cite{qemu} to gather the translation blocks executed at runtime along with their number of executions.
For SPEC, we used static binary analysis tools to extract the basic blocks along with performance counters statistics in order to recover the performance-critical section of the code, as the cost of running an emulator was too high to reproduce Polybench's setup. Overall these two benchmark suites generate thousands of basic blocks, and for each we use the various methods above to display the predicted performance of a microkernel made of the same instruction mix that is occurring in that basic block. This evaluation approach allows to generate a high variety of realistic instruction mixes (e.g., combining SIMD and address calculations for numerical kernels like in Polybench).

Fig.~\ref{fig:realcpu_results} synthesizes our results in two pieces. First, Fig.~\ref{fig:realcpu_heatmap} displays the results as a heatmap for each basic block, comparing the predicted throughput with the measured one. A dark area at coordinate $(x, y)$ means that the selected tool has a prediction accuracy of $y$ for a significant number of microkernels with a real IPC of $x$.

Then, Table~\ref{tab:realcpu_results} synthesizes, for each tool, its error rate, aggregated over all the basic blocks of the test suite using a \textit{Root-Mean-Square} method:
\[ \text{Err}_\text{RMS, tool} = \sqrt{\sum_{i}
    \frac{\text{weight}_i}{\sum_j \text{weight}_j} \left(
    \frac{\text{IPC}_{i,\text{tool}} - \text{IPC}_{i,\text{native}}}{\text{IPC}_{i,\text{native}}}
    \right)^2
    }
\]

We also provide Kendall's $\tau$ coefficient~\cite{kendall1938tau}, a measure of the rank correlation of a predictor~-- that is, for each pair of basic blocks, whether a predictor predicted correctly which block had the higher IPC\@. The coefficient varies between $-1$ (full anti-correlation) and $1$ (full correlation).

The same table also provides a \textit{coverage} metric, \emph{with respect to \tool}. This metric characterizes the proportion of basic blocks supported by \tool{} that the tool was able to process. Note that the ability to process a basic block varies from tool to tool: some work in degraded mode when meeting instructions they cannot handle, some will crash on the basic block. For PMEvo, we ignored any instruction not supported by their mapping --~degrading the quality of the result; hence, a plain error is a basic block in which \emph{no single instruction was supported}. Although it would be fairer to other tools to measure absolute coverage --~that is, the proportion of basic blocks supported by the tool, regardless of what \tool supports~--, technical limitations prevented us from doing so: running the various tools requires our back-end to generate assembly code, which can only be done for the instructions it supports.

\begin{table}[h!tb]
{\small
%\vspace{-.2cm}
\begin{center}
\caption{\sc\label{fig:test_environment}Main features of the obtained mappings}
%\vspace{-.2cm}
\begin{tabular}{c|c|c}
\toprule
Machine & SKL-SP & ZEN1 \\
\midrule
\multirow{2}{*}{Processor} & 2x Intel Xeon & AMD EPYC \\
 & Silver 4114 & 7401P \\
Cores & 20 & 24 \\
\midrule
Benchmarking time & 8h & 6h \\
LP solving time & 2h & 2h \\
Overall time & 10h & 8h \\
\midrule
Gen. microbenchmarks & $\sim$ 1,000,000 & $\sim$ 1,000,000 \\
Resources found & 17 & 17 \\
uops' inst. supported & 3313 & 1104 \\
Instructions mapped & 2586 & 2596 \\
\bottomrule
\end{tabular}
\end{center}
}
\vspace{-.25cm}
\end{table}

%We compare the number of instructions supported by \tool{} with the ones supported by uops.info as a baseline, but, as uops supports only partially AMD's architecture (providing only throughput and latencies, but no usable port mapping), less than half the instructions supported by our tool are present.
We compare the number of instructions supported by \tool{} with the ones supported by uops.info as a baseline. As uops supports only partially AMD's architecture (providing only throughput and latencies, but no usable port mapping), less than half the instructions supported by our tool are present for this target.
Contrarily, on SKL-SP, uops supports the AVX-512 extension, therefore leading to a more complete set of supported instructions.
PMEvo's mapping behaves poorly in terms of coverage (see Fig.~\ref{tab:realcpu_results}), failing to support all instructions in more than 25 \% of the basic blocks on any benchmark and processor tested. This behavior is due to our different compilation options, as PMEvo's supported instructions are directly collected from their SPEC2017 binaries. As a consequence, both MSE and Kendall's tau values are lower than other tools as those unsupported instructions are treated as if they took no resource at all on our IPC estimates.

Moreover, \tool requires 2h of solving time (see Tbl.~\ref{fig:test_environment}) to map about 2500 instructions. This is between one half (SKL-SP) and one eighth (Zen) of PMEvo's solving time~\cite{uops.info}, demonstrating the scalability of \tool with respect to the number of instructions.

In Fig.~\ref{fig:realcpu_results}, we observe that \tool{} performs significantly better than uops.info and PMEvo on both platforms. On Skylake, it outperforms all other tested tools in terms of Kendall's tau, and compares well with IACA and LLVM-MCA, archiving sub-10\,\% mean square error rate on SPEC2017. However, those two last tools use manual expertise and are tailored for a platform, whereas our tool is fully automated and generic.

On Zen1, \tool{} is comparable to LLVM-MCA, but shows a greater error rate than on Intel. This is due to the internal organization of the Zen microarchitecture, which uses a separated pipeline for integer/control flow and floating point/vector operations.
As \tool{} tries to minimize the number of resources, this separation is not properly detected, leading to IPC predictions lower than the actual value as seen on the heatmaps on Fig.~\ref{fig:realcpu_heatmap}.

More generally, IACA, uops.info and LLVM-MCA tend to over-estimate the IPC, which is due to their port-based approach: bottlenecks coming from neither ports nor front-end limitations are not taken into account, leading to higher IPC estimations for microkernels where other resources are bottlenecking. Contrarily, benchmarking-based approaches (\tool and PMEvo) present both under and over approximations as they are based on real-life execution, where all bottlenecks are present. Note that \tool, IACA, LLVM-MCA (Zen1 only) and PMEvo (Zen only) also express the front-end bottleneck: the limit on the maximal number of instructions being decoding in one cycle (no over-approximation of microkernels with high IPC), that is, a maximal IPC of 4 on SKL-SP and 5 for Zen1. Therefore, we expect \tool (and PMEvo) to have maximal error rate on benchmarks with few instructions, case in which some undetected / wrongly detected common resource will have higher importance, whereas LLVM-MCA, uops.info and IACA will tend to be more fragile on long microkernels with possible non-port related resources~-- especially memory ones.

%In Fig.~\ref{fig:realcpu_results} we observe that \tool{} performs significantly better than PMEvo and uops.info on both platforms. On Skylake, it compares well with IACA, showing similar error rate and even a higher Kendall's Tau, and outperforms LLVM-MCA whereas those two last tools use manual expertise and are tailored for a platform. On Zen1, \tool{} is comparable to LLVM-MCA.

\section{Conclusion}

We presented \tool{} which automatically builds a resource mapping for CPU instructions.
This allows to model not only execution port usage, but also other limiting resources, such as the front-end or the reorder buffer.
We presented an end-to-end approach to enable the mapping of thousands of instructions in a few hours, including microbenchmarking time.
Our key contributions include the mathematically rigorous formulation of the port mapping problem as solving iteratively linear programs, enabling an incremental and scalable approach to handling thousands of instructions.
We provided a method to automatically generate microbenchmarks saturating specific resources, alleviating the need for statistical sampling.
We demonstrated on one Intel and one AMD high-performance CPUs that \tool{} generates automatically practical port mappings that compare favorably with state-of-the-art systems like IACA and uops.info that either use performance counters or manual expertise.

\section*{Acknowledgments}
The works have been funded by ECSEL-JU under the program ECSEL-Innovation Actions-2018 (ECSEL-IA) for research project CPS4EU (ID-826276) in the area Cyber-Physical Systems. This work was also supported in part by the U.S. National Science Foundation award CCF-1750399.

%\section*{Artifact}
%\input{artifact.tex}

\bibliographystyle{IEEEtranS}
\bibliography{main}

% Generated by IEEEtranS.bst, version: 1.14 (2015/08/26)
\begin{thebibliography}{10}
\providecommand{\url}[1]{#1}
\csname url@samestyle\endcsname
\providecommand{\newblock}{\relax}
\providecommand{\bibinfo}[2]{#2}
\providecommand{\BIBentrySTDinterwordspacing}{\spaceskip=0pt\relax}
\providecommand{\BIBentryALTinterwordstretchfactor}{4}
\providecommand{\BIBentryALTinterwordspacing}{\spaceskip=\fontdimen2\font plus
\BIBentryALTinterwordstretchfactor\fontdimen3\font minus
  \fontdimen4\font\relax}
\providecommand{\BIBforeignlanguage}[2]{{%
\expandafter\ifx\csname l@#1\endcsname\relax
\typeout{** WARNING: IEEEtranS.bst: No hyphenation pattern has been}%
\typeout{** loaded for the language `#1'. Using the pattern for}%
\typeout{** the default language instead.}%
\else
\language=\csname l@#1\endcsname
\fi
#2}}
\providecommand{\BIBdecl}{\relax}
\BIBdecl

\bibitem{PMD-full}
\BIBentryALTinterwordspacing
 [Online]. Available:
  \url{https://www.dropbox.com/s/zvsnj4wsx0fj775/PMD-full.pdf?dl=0}
\BIBentrySTDinterwordspacing

\bibitem{qemu}
``{QEMU}: the {FAST!} processor emulator,'' \url{https://www.qemu.org}.

\bibitem{DBLP:journals/corr/abs-1911-03282}
\BIBentryALTinterwordspacing
A.~Abel and J.~Reineke, ``nano{B}ench: A low-overhead tool for running
  microbenchmarks on x86 systems,'' \emph{arXiv e-prints}, vol. abs/1911.03282,
  2019. [Online]. Available: \url{http://arxiv.org/abs/1911.03282}
\BIBentrySTDinterwordspacing

\bibitem{uops.info}
\BIBentryALTinterwordspacing
------, ``uops.info: Characterizing latency, throughput, and port usage of
  instructions on intel microarchitectures,'' in \emph{Proceedings of the
  Twenty-Fourth International Conference on Architectural Support for
  Programming Languages and Operating Systems, {ASPLOS} 2019}, I.~Bahar,
  M.~Herlihy, E.~Witchel, and A.~R. Lebeck, Eds.\hskip 1em plus 0.5em minus
  0.4em\relax New York, NY, USA: {ACM}, April 2019, pp. 673--686. [Online].
  Available: \url{https://doi.org/10.1145/3297858.3304062}
\BIBentrySTDinterwordspacing

\bibitem{uiCA}
------, ``Accurate throughput prediction of basic blocks on recent intel
  microarchitectures,'' 2021.

\bibitem{MCSim}
\BIBentryALTinterwordspacing
J.~H. Ahn, S.~Li, S.~O, and N.~P. Jouppi, ``{McSimA+}: A manycore simulator
  with application-level+ simulation and detailed microarchitecture modeling,''
  in \emph{2012 {IEEE} International Symposium on Performance Analysis of
  Systems and Software}.\hskip 1em plus 0.5em minus 0.4em\relax Austin, TX,
  USA: {IEEE} Computer Society, April 2013, pp. 74--85. [Online]. Available:
  \url{https://doi.org/10.1109/ISPASS.2013.6557148}
\BIBentrySTDinterwordspacing

\bibitem{SPECCPU2017}
\BIBentryALTinterwordspacing
J.~Bucek, K.~Lange, and J.~von Kistowski, ``{SPEC CPU2017}: Next-generation
  compute benchmark,'' in \emph{Companion of the 2018 {ACM/SPEC} International
  Conference on Performance Engineering, {ICPE} 2018}, K.~Wolter, W.~J.
  Knottenbelt, A.~van Hoorn, and M.~Nambiar, Eds.\hskip 1em plus 0.5em minus
  0.4em\relax {ACM}, April 2018, pp. 41--42. [Online]. Available:
  \url{https://doi.org/10.1145/3185768.3185771}
\BIBentrySTDinterwordspacing

\bibitem{EXEgesis}
\BIBentryALTinterwordspacing
C.~Chatelet, C.~Courbet, O.~Sykora, and N.~Paglieri. Google exegesis. [Online].
  Available: \url{https://llvm.org/docs/CommandGuide/llvm-exegesis.html}
\BIBentrySTDinterwordspacing

\bibitem{MemH-jack}
C.~L. {Coleman} and J.~W. {Davidson}, ``Automatic memory hierarchy
  characterization,'' in \emph{2001 IEEE International Symposium on Performance
  Analysis of Systems and Software. ISPASS.}, 2001, pp. 103--110.

\bibitem{Intel-manual}
\BIBentryALTinterwordspacing
I.~Corporation. Intel 64 and ia-32 architectures optimization reference manual.
  [Online]. Available:
  \url{https://www.intel.com/content/dam/doc/manual/64-ia-32-architectures-optimization-manual.pdf}
\BIBentrySTDinterwordspacing

\bibitem{Intel-XED}
\BIBentryALTinterwordspacing
------. Intel x86 encoder decoder (intel xed). [Online]. Available:
  \url{https://github.com/intelxed/xed}
\BIBentrySTDinterwordspacing

\bibitem{DBLP:conf/hipc/DjoudiNJ08}
\BIBentryALTinterwordspacing
L.~Djoudi, J.~Noudohouenou, and W.~Jalby, ``The design and architecture of
  {MAQAOAdvisor}: A live tuning guide,'' in \emph{Proceedings of the 15th
  International Conference on High Performance Computing}, ser. {HiPC} 2008,
  P.~Sadayappan, M.~Parashar, R.~Badrinath, and V.~K. Prasanna, Eds., vol.
  5374.\hskip 1em plus 0.5em minus 0.4em\relax Berlin, Heidelberg:
  Springer-Verlag, December 2008, pp. 42--56. [Online]. Available:
  \url{https://doi.org/10.1007/978-3-540-89894-8\_8}
\BIBentrySTDinterwordspacing

\bibitem{AgnerFog}
\BIBentryALTinterwordspacing
A.~Fog. (2020) Instruction tables: Lists of instruction latencies, through-puts
  and micro-operation breakdowns for intel, {AMD} and {VIA} {CPU}s. [Online].
  Available: \url{http://www.agner.org/optimize/instruction_tables.pdf}
\BIBentrySTDinterwordspacing

\bibitem{SPIRAL}
\BIBentryALTinterwordspacing
F.~Franchetti, T.~M. Low, D.~Popovici, R.~M. Veras, D.~G. Spampinato, J.~R.
  Johnson, M.~P{\"{u}}schel, J.~C. Hoe, and J.~M.~F. Moura, ``{SPIRAL}: Extreme
  performance portability,'' \emph{Proceedings of the {IEEE}}, vol. 106,
  no.~11, pp. 1935--1968, 2018. [Online]. Available:
  \url{https://doi.org/10.1109/JPROC.2018.2873289}
\BIBentrySTDinterwordspacing

\bibitem{Granlund}
\BIBentryALTinterwordspacing
T.~Granlund. (2017) Instruction latencies and throughput for {AMD} and intel
  x86 processors. [Online]. Available:
  \url{https://gmplib.org/~tege/x86-timing.pdf}
\BIBentrySTDinterwordspacing

\bibitem{DBLP:journals/corr/HammerEHW17}
J.~Hammer, J.~Eitzinger, G.~Hager, and G.~Wellein, ``Kerncraft: A tool for
  analytic performance modeling of loop kernels,'' in \emph{Tools for High
  Performance Computing 2016}, vol. abs/1702.04653.\hskip 1em plus 0.5em minus
  0.4em\relax Cham: Springer International Publishing, 2017, pp. 1--22.

\bibitem{IACA}
\BIBentryALTinterwordspacing
I.~Hirsh and G.~S. Intel® architecture code analyzer. [Online]. Available:
  \url{https://software.intel.com/en-us/articles/intel-architecture-code-analyzer}
\BIBentrySTDinterwordspacing

\bibitem{instlatx64.atw.hu}
\BIBentryALTinterwordspacing
instlatx64. x86, x64 instruction latency, memory latency and cpuid dumps.
  [Online]. Available: \url{http://instlatx64.atw.hu/}
\BIBentrySTDinterwordspacing

\bibitem{kendall1938tau}
M.~G. Kendall, ``A new measure of rank correlation,'' \emph{Biometrika},
  vol.~30, no. 1/2, pp. 81--93, 1938.

\bibitem{LLVM}
\BIBentryALTinterwordspacing
C.~Lattner and V.~S. Adve, ``{LLVM}: A compilation framework for lifelong
  program analysis {\&} transformation,'' in \emph{2nd {IEEE}/{ACM}
  International Symposium on Code Generation and Optimization {(CGO}
  2004)}.\hskip 1em plus 0.5em minus 0.4em\relax San Jose, CA, USA: {IEEE}
  Computer Society, March 2004, pp. 75--88. [Online]. Available:
  \url{https://doi.org/10.1109/CGO.2004.1281665}
\BIBentrySTDinterwordspacing

\bibitem{OSACA}
J.~Laukemann, J.~Hammert, J.~Hofmann, G.~Hager, and G.~Wellein, ``Automated
  instruction stream throughput prediction for intel and {AMD}
  microarchitectures,'' in \emph{2018 {IEEE}/{ACM} Performance Modeling,
  Benchmarking and Simulation of High Performance Computer Systems
  ({PMBS})}.\hskip 1em plus 0.5em minus 0.4em\relax Dallas, TX, USA: {IEEE}
  Computer Society, {ACM}, November 2018, pp. 121--131.

\bibitem{Zesto}
\BIBentryALTinterwordspacing
G.~H. Loh, S.~Subramaniam, and Y.~Xie, ``Zesto: A cycle-level simulator for
  highly detailed microarchitecture exploration,'' in \emph{{IEEE}
  International Symposium on Performance Analysis of Systems and Software,
  {ISPASS} 2009}.\hskip 1em plus 0.5em minus 0.4em\relax Boston, Massachusetts,
  USA: {IEEE} Computer Society, April 2009, pp. 53--64. [Online]. Available:
  \url{https://doi.org/10.1109/ISPASS.2009.4919638}
\BIBentrySTDinterwordspacing

\bibitem{GEM5}
\BIBentryALTinterwordspacing
J.~Lowe{-}Power, A.~M. Ahmad, A.~Akram, M.~Alian, R.~Amslinger, M.~Andreozzi,
  A.~Armejach, N.~Asmussen, S.~Bharadwaj, G.~Black, G.~Bloom, B.~R. Bruce,
  D.~R. Carvalho, J.~Castrill{\'{o}}n, L.~Chen, N.~Derumigny, S.~Diestelhorst,
  W.~Elsasser, M.~Fariborz, A.~F. Farahani, P.~Fotouhi, R.~Gambord, J.~Gandhi,
  D.~Gope, T.~Grass, B.~Hanindhito, A.~Hansson, S.~Haria, A.~Harris, T.~Hayes,
  A.~Herrera, M.~Horsnell, S.~A.~R. Jafri, R.~Jagtap, H.~Jang, R.~Jeyapaul,
  T.~M. Jones, M.~Jung, S.~Kannoth, H.~Khaleghzadeh, Y.~Kodama, T.~Krishna,
  T.~Marinelli, C.~Menard, A.~Mondelli, T.~M{\"{u}}ck, O.~Naji, K.~Nathella,
  H.~Nguyen, N.~Nikoleris, L.~E. Olson, M.~S. Orr, B.~Pham, P.~Prieto,
  T.~Reddy, A.~Roelke, M.~Samani, A.~Sandberg, J.~Setoain, B.~Shingarov, M.~D.
  Sinclair, T.~Ta, R.~Thakur, G.~Travaglini, M.~Upton, N.~Vaish, I.~Vougioukas,
  Z.~Wang, N.~Wehn, C.~Weis, D.~A. Wood, H.~Yoon, and {\'{E}}.~F. Zulian, ``The
  gem5 simulator: Version 20.0+,'' 2020. [Online]. Available:
  \url{https://arxiv.org/abs/2007.03152}
\BIBentrySTDinterwordspacing

\bibitem{MIAMI}
\BIBentryALTinterwordspacing
G.~Marin, J.~J. Dongarra, and D.~Terpstra, ``{MIAMI}: A framework for
  application performance diagnosis,'' in \emph{2014 {IEEE} International
  Symposium on Performance Analysis of Systems and Software, {ISPASS}
  2014}.\hskip 1em plus 0.5em minus 0.4em\relax Monterey, CA, USA: {IEEE}
  Computer Society, March 2014, pp. 158--168. [Online]. Available:
  \url{https://doi.org/10.1109/ISPASS.2014.6844480}
\BIBentrySTDinterwordspacing

\bibitem{Ithemal}
\BIBentryALTinterwordspacing
C.~Mendis, A.~Renda, S.~P. Amarasinghe, and M.~Carbin, ``Ithemal: Accurate,
  portable and fast basic block throughput estimation using deep neural
  networks,'' in \emph{Proceedings of the 36th International Conference on
  Machine Learning, {ICML} 2019}, ser. Proceedings of Machine Learning
  Research, K.~Chaudhuri and R.~Salakhutdinov, Eds., vol.~97.\hskip 1em plus
  0.5em minus 0.4em\relax Long Beach, California, {USA}: {PMLR}, June 2019, pp.
  4505--4515. [Online]. Available:
  \url{http://proceedings.mlr.press/v97/mendis19a.html}
\BIBentrySTDinterwordspacing

\bibitem{hclust}
F.~Nielsen, \emph{Hierarchical Clustering}, 02 2016, pp. 195--211.

\bibitem{polybench}
L.-N. Pouchet and T.~Yuki, ``{PolyBench/C}: The polyhedral benchmark suite,
  version 4.2,'' 2016, \url{http://polybench.sf.net}.

\bibitem{GCC}
\BIBentryALTinterwordspacing
G.~C. Project. (1987) {GNU} compiler collection (gcc). [Online]. Available:
  \url{https://gcc.gnu.org/}
\BIBentrySTDinterwordspacing

\bibitem{PMEvo}
\BIBentryALTinterwordspacing
F.~Ritter and S.~Hack, ``Pmevo: portable inference of port mappings for
  out-of-order processors by evolutionary optimization,'' in \emph{Proceedings
  of the 41st {ACM} {SIGPLAN} International Conference on Programming Language
  Design and Implementation, {PLDI} 2020}, A.~F. Donaldson and E.~Torlak,
  Eds.\hskip 1em plus 0.5em minus 0.4em\relax New York, USA: {ACM}, June 2020,
  pp. 608--622. [Online]. Available:
  \url{https://doi.org/10.1145/3385412.3385995}
\BIBentrySTDinterwordspacing

\bibitem{CQA}
\BIBentryALTinterwordspacing
A.~C. Rubial, E.~Oseret, J.~Noudohouenou, W.~Jalby, and G.~Lartigue, ``{CQA}: A
  code quality analyzer tool at binary level,'' in \emph{21st International
  Conference on High Performance Computing, {HiPC} 2014}.\hskip 1em plus 0.5em
  minus 0.4em\relax Goa, India: {IEEE} Computer Society, December 2014, pp.
  1--10. [Online]. Available: \url{https://doi.org/10.1109/HiPC.2014.7116904}
\BIBentrySTDinterwordspacing

\bibitem{DBLP:conf/hipc/RubialONJL14}
\BIBentryALTinterwordspacing
------, ``{CQA}: A code quality analyzer tool at binary level,'' in \emph{21st
  International Conference on High Performance Computing, {HiPC} 2014}.\hskip
  1em plus 0.5em minus 0.4em\relax Goa, India: {IEEE} Computer Society,
  December 2014, pp. 1--10. [Online]. Available:
  \url{https://doi.org/10.1109/HiPC.2014.7116904}
\BIBentrySTDinterwordspacing

\bibitem{ZSim}
\BIBentryALTinterwordspacing
D.~S{\'{a}}nchez and C.~Kozyrakis, ``{ZSim}: fast and accurate
  microarchitectural simulation of thousand-core systems,'' in \emph{40th
  Annual International Symposium on Computer Architecture, ({ISCA}'13)},
  A.~Mendelson, Ed.\hskip 1em plus 0.5em minus 0.4em\relax New York, NY, USA:
  {ACM}, June 2013, pp. 475--486. [Online]. Available:
  \url{https://doi.org/10.1145/2485922.2485963}
\BIBentrySTDinterwordspacing

\bibitem{LLVM:MCA}
\BIBentryALTinterwordspacing
{Sony Corporation} and L.~Project. {LLVM} machine code analyzer. [Online].
  Available: \url{https://llvm.org/docs/CommandGuide/llvm-mca.html}
\BIBentrySTDinterwordspacing

\bibitem{Agner:UsedByIntel}
\BIBentryALTinterwordspacing
C.~Topper. (2018, Mar.) Update to the llvm scheduling model for intel sandy
  bridge, haswell, broadwell, and skylake processors. [Online]. Available:
  \url{https://github.com/llvm/llvm-project/commit/cdfcf8ecda8065fda495d73ed16277668b3b56dc}
\BIBentrySTDinterwordspacing

\bibitem{DBLP:conf/icppw/TreibigHW10}
\BIBentryALTinterwordspacing
J.~Treibig, G.~Hager, and G.~Wellein, ``{LIKWID}: A lightweight
  performance-oriented tool suite for x86 multicore environments,'' in
  \emph{39th International Conference on Parallel Processing ({ICPP}) Workshops
  2010}, W.~Lee and X.~Yuan, Eds.\hskip 1em plus 0.5em minus 0.4em\relax San
  Diego, California, USA: {IEEE} Computer Society, September 2010, pp.
  207--216. [Online]. Available: \url{https://doi.org/10.1109/ICPPW.2010.38}
\BIBentrySTDinterwordspacing

\bibitem{Roofline}
\BIBentryALTinterwordspacing
S.~Williams, A.~Waterman, and D.~Patterson, ``Roofline: An insightful visual
  performance model for multicore architectures,'' \emph{Commun. ACM}, vol.~52,
  no.~4, pp. 65--76, Apr. 2009. [Online]. Available:
  \url{http://doi.acm.org/10.1145/1498765.1498785}
\BIBentrySTDinterwordspacing

\bibitem{PTLSim}
\BIBentryALTinterwordspacing
M.~T. Yourst, ``{PTLsim}: A cycle accurate full system x86-64
  microarchitectural simulator,'' in \emph{2007 {IEEE} International Symposium
  on Performance Analysis of Systems and Software}.\hskip 1em plus 0.5em minus
  0.4em\relax San Jose, California, USA: {IEEE} Computer Society, April 2007,
  pp. 23--34. [Online]. Available:
  \url{https://doi.org/10.1109/ISPASS.2007.363733}
\BIBentrySTDinterwordspacing

\bibitem{BLIS}
\BIBentryALTinterwordspacing
F.~G.~V. Zee and R.~A. van~de Geijn, ``{BLIS:} a framework for rapidly
  instantiating {BLAS} functionality,'' \emph{{ACM} Transactions on
  Mathematical Software}, vol.~41, no.~3, June 2015. [Online]. Available:
  \url{https://doi.org/10.1145/2764454}
\BIBentrySTDinterwordspacing

\end{thebibliography}

%%% ===== APPENDIX: prooves of "Completness of the mapping" ========
% ONLY FOR THE FULL VERSION
\iffulldoc
  \newpage
  ~
  \newpage
  \appendix
  % Full proofs and definitions from Section 3 goes there.
\section{Appendix - Proof of the Equivalence of the Dual Eepresentation}
\label{app:dual_repr_equiv}

This section contains the definition of a dual conjunctive bipartite port mapping, to a disjunctive port mapping.
It also contains the full proof of equivalence between these two representations.

\subsection{Primary Definitions}
\label{sec:defs}

% === µ Kernel DEFINITION
\begin{restatable}[Microkernel]{definition}{defmicrokernel}
\label{def:microkernel}
A microkernel $K$ is an \emph{infinite loop} made up of a finite multiset of instructions, $K=I_1^{\sigma_{K,1}}I_2^{\sigma_{K,2}}\cdots I_m^{\sigma_{K,m}}$ without dependencies between instructions. The number of instructions executed during one loop iteration is $|K|=\sum_{i}\sigma_{K,i}$.
\end{restatable}
\vspace{-.2cm}

% === Disjunctive port mapping DEFINITION
\begin{restatable}[Disjunctive port mapping]{definition}{defdisjunctive}
\label{def:disjunctive}
A disjunctive port mapping is a bipartite graph $(V,\Resources,E)$ where:
$V$ represents the set of \muops;
$\Resources$ represents the set of resources (corresponding to execution ports in a real-world CPU);
$E \subset V\times \Resources$ expresses the possible mappings from \muops to ports.
In this original form each port $r\in\Resources$ has a throughput $\rho(r)$ of 1.

Let $K=I_1^{\sigma_{K,1}}I_2^{\sigma_{K,2}}\cdots I_m^{\sigma_{K,m}}$ be a microkernel where each instruction is composed of a single \muop $v_i$.

A valid assignment represents the choice of which resources to associate with a given instance of an instruction. However, this choice might change between iterations. Thus, we represent the valid assignment as a mapping $p : \Instructions \times \Resources \mapsto [0;1]$ where $p_{i,r}$ corresponds to the frequency a given resource is chosen.
We also define $R_i(p) = \{r,\ p_{i,r} \neq 0\}$. This assignment is valid if:
$$\forall I_i\in K, \forall r \in R_i(p),\  \left(v_i, r\right)\in E $$
$$\forall I_i\in K,\ \sum_{r \in R_i(p)} p_{i,r} = 1$$

The \emph{execution time} of an assignment $(p_{i,r})_{i,r}$, is:
$$t_\textit{end}=\max_{r\in \Resources} \sum_{i \in K} \sigma_{K, i} \cdot p_{i,r}$$

The minimal execution time over all valid assignments is denoted $t(K)$ (obtained using an \emph{optimal assignment}).
\end{restatable}

% === Conjunctive port mapping DEFINITION
\begin{restatable}[Conjunctive port mapping]{definition}{defconjunctive}
\label{def:conjunctive}
A conjunctive port mapping is a bipartite \emph{weighted} graph $(I,\Resources,E,$ $\rho_{I,\Resources})$ where:
$I$ represents the set of instructions;
$\Resources$ represents the set of \emph{abstract resources};
$E \subset I\times \Resources$ expresses the required use of abstract resources for each instruction;

Each abstract resource $r\in \Resources$ has a (normalized) throughput of 1;
An instruction $i$ that uses a resource $r$ ($(i,r)\in E$) always uses the same proportion (number of cycles, possibly lower/greater than 1) $\rho_{i,r}\in \Q^+$;
If $i$ does not use $r$, then $\rho_{i,r}=0$.

Let $K=I_1^{\sigma_{K,I_1}}I_2^{\sigma_{K,I_2}}\cdots I_m^{\sigma_{K,I_m}}$ be a microkernel.
In a steady state execution of $K$, for each loop iteration, instruction $i$ must use resource $r$ $(\sigma_{K,i}\rho_{i,r})$ cycles.

The number of cycles required to execute one loop iteration is:
$$t(K)=\max_{r\in \Resources} \left(\sum_{i\in K} \sigma_{K,i}\cdot\rho_{i,r}\right)$$
\end{restatable}

One should observe that Def.~\ref{def:conjunctive} defines formally a \emph{normalized} version of the graph used in the illustrative example of Sec.~\ref{subsec:dual}; where throughputs of abstract resources are set to 1.
For the sake of clarity, we used non-normalized throughputs (that is, different than 1) in Fig.~\ref{fig:resource-mapping} with the following notations: \emph{use} stands for the non-normalized usage, and \emph{load} for the normalized $\rho_{i,r}$, equal to $\frac{\#use_i}{throughput(r)}$. For example, \vcv{} uses 2 times $r_{01}$, which has a throughput of 2: its load $\rho_{\vcv,r_{01}}$ is equal to $1$.
Similarly, $\rho_{\addss,r_{016}}=1/3$.

\begin{restatable}[Throughput]{definition}{defthroughput}
\label{def:throughput}
The throughput $\ipc{K}$ of a microkernel $K$ is its instruction per cycle rate (IPC), defined as:
\begin{align*}
    \bullet \qquad \ipc{K} =  \frac{|K|}{t(K)} = & \max_{\textrm{valid assignment $p$}}\left(\frac{\sum_{i\in K} \sigma_{K,i}}{\max_{r\in \Resources} \sum_{i \in K} \sigma_{K, i} \cdot p_{i,r}}\right) \\
    &\qquad\text{for a disjunctive port mapping.}\\
    \bullet \qquad \ipc{K} = \frac{|K|}{t(K)} = & \frac{\sum_{i\in K} \sigma_{K,i}}{\max_{r\in \Resources} \sum_{i\in K} \sigma_{K,i}\cdot\rho_{i,r}} \\
    & \qquad \text{for a conjunctive port mapping}
\end{align*}
\end{restatable}
\vspace{-0.5cm}

\paragraph*{\bf Example} Given $a$ and $b$ two instructions, $a^{\ipc{a}}b^{\ipc{b}}$ represents a microkernel repeating $a$ and $b$ as many times as their respective IPC $\ipc{a}$ and $\ipc{b}$.
We note its throughput $\ipc{a^{\ipc{a}}b^{\ipc{b}}}$.

% === Duality DEFINITION
% ONLY USED IN THE COMPLETE VERSION
\begin{restatable}[$\mathcal r$-dual conjunctive port mapping]{definition}{defdual}
\label{def:dual}
Let $(V,$ $\Resources,E)$ be a disjunctive port mapping.
Let $\mathcal r$ be a non-empty set of subsets of $\Resources$.
We define its $\mathcal r$-dual, a conjunctive port mapping, as $(V,\overline{\Resources},\overline{E})$ such that:
$$\begin{array}{c}
\overline{\Resources} = \left\{ \overline{r}_J,\ J\in {\mathcal r} \right\}\\
\overline{E} = \left\{ (v,\overline{r}_J) \textrm{ s.t. } \{ r, (v,r)\in E \} \subseteq J \right\}\\
\rho(\overline{r}_J)= \sum_{r_j\in J}\rho(r_j)=\left| J \right|
\end{array}$$

% Normalisation
Then, we can normalize this graph by adding weights to edges, and update the resource throughput, noted $\rho^N$ as it is normalized.
$$\begin{array}{l}
\rho^N_{i,\overline{r}_J} = \left\{\begin{array}{ll}
    1 / \rho(\overline{r}_J) & \textrm{ if } (i,\overline{r}_J)\in\overline{E}\\
    0  & \textrm{ else}
\end{array}\right.\\
\rho^N(\overline{r}_J) = 1
\end{array}$$
\end{restatable}

\subsection{Equivalence between Disjunctive and Conjunctive formulations}
This section provides the main intuition to understand the equivalence between the disjunctive and the conjunctive form.

% === Saturated set DEFINITION
\begin{restatable}[Saturated port set]{definition}{defsaturating}
  \label{def:saturating}
  Consider a microkernel $K$. Let $(p_{i,r})_{i,r}$ be a valid assignment of $K$ for a disjunctive port mapping $(V,\Resources, E)$. Its saturated port set $\mathcal{S}$ is defined as follows:
$$\mathcal{S}=\left\{r_s \textrm{ such that }  t_\textit{end} = \sum_{i \in K} \sigma_{K, i} \cdot p_{i,r_s} \right\}$$

That is, resources $r_S$ for which their loads $\sum_{i \in K} \sigma_{K, i} \cdot p_{i,r_s}$ correspond to $|K|/\ipc{K}$, the steady state execution time of $K$.
\end{restatable}

% === Lemma saturated set connections
% ONLY USED IN THE COMPLETE VERSION
\begin{restatable}[Saturated set assumption]{lemma}{SaturatingS}
\label{lem:saturatingS}
Let $(p_{i,r})_{i,r}$ be a valid assignment for a microkernel $K$ in a disjunctive port mapping $(V,\Resources, E)$ and $\mathcal{S}$ its saturated set. Let $r_s$ and $r_t$ be two resources such that $(v,r_s) \in E$ and $(v,r_t) \in E$, we assume $r_s\in \mathcal{S}$ and $r_t \not\in \mathcal{S}$.

Then, either there exists a faster valid assignment for which both resources $r_s$ and $r_t$ are saturated, or there exists a valid assignment whose saturated set is strictly smaller than $p$.
\end{restatable}

A direct consequence of this lemma is:

% ONLY USED IN THE COMPLETE VERSION
\begin{restatable}[Saturating assignment]{corollary}{saturatingA}
\label{cor:saturatingA}
Let us consider an optimal assignment $(p_{i,r})_{i,r}$ of a list of \muops $K$ on a disjunctive port mapping $(V, \Resources , E)$, such that the size of its saturated set $\mathcal{S}$ is minimal. For all $v\in V$ such that there are $(r_x, r_y) \in \Resources^2$ connected to $v$ (i.e. $(v,r_x) \in E$ and $(v, r_y) \in E$): if $r_x\in\mathcal{S}$, then $r_y \in \mathcal{S}$.

Thus:
$   \forall i \in \Instructions,\ \left[ R_i(p) \subset \mathcal{S} \Leftrightarrow \{r,\ (v_i,r)\in E\} \subset \mathcal{S} \right]$
\end{restatable}

% === Theorem duality
% ONLY USED IN THE COMPLETE VERSION
\begin{restatable}[Equivalence of $\mathcal{r}$-duality]{theorem}{thmdual_real}
\label{thm:dual_real}
  Let $K$ be a microkernel.
Let $(V,\Resources,E)$ (with the set of resources $\Resources=\{r_j\}_j$), $\mathcal r$ a set of subsets of $\Resources$, and $(V,\overline{\Resources},\overline{E})$ (with the set of resources $\overline{\Resources}$ also denoted $\{\overline{r}_J\}_{J\in {\mathcal r}}$) be its ${\mathcal r}$-dual.

(i) Let $(p_{i,r})_{i,r}$ be a valid optimal assignment (i.e. of minimal execution time and minimal saturated set size) of $K$ with regard to $(V,\Resources,E)$. This assignment can be translated into its $\mathcal{r}$-dual, with no change to its execution time. In other words, $\overline{t}(K)\leq t(K)$.

(ii) If $\mathcal r$ is the set of \emph{all} subsets of $\Resources$ then $\overline{t}(K)=t(K)$.
\end{restatable}

\begin{restatable}[Equivalence]{theorem}{thmdual}
\label{thm:dual}Let $K$ be a microkernel and $(V,\Resources,E)$ (with the associated throughput function $t$) be a disjunctive port mapping. Then, there exists $\mathcal r$ a set of subsets of $\Resources$, and $(V,\overline{\Resources},\overline{E})$ a conjunctive port mapping called the \emph{dual} (with the associated throughput $\overline{t}$) whose set of resources $\overline{\Resources}$ is indexed by ${\mathcal r}$ such that:

(i) Every $(p_{i,r})_{i,r}$ optimal assignment (i.e. of minimal execution time and minimal saturated set size) of $K$ with regard to $(V,\Resources,E)$ can be translated into $(V,\overline{\Resources},\overline{E})$, with no change of its execution time. In other words, $\overline{t}(K)\leq t(K)$.

(ii) If $\mathcal r$ is the set of \emph{all} subsets of $\Resources$ then $\overline{t}(K)=t(K)$.
\end{restatable}

% === Proof
% ONLY USED IN THE COMPLETE VERSION
\begin{proof}
    (i) Let $p$ be an optimal valid assignment for a list of \muops $K$ on a disjunctive port mapping $(V, \Resources, E)$, which minimizes its saturated set size $|\mathcal{S}|$.

    From the definition of an execution time, we have: $$\forall r \in \Resources,\ \sum_{i\in K} \sigma_{K,i} \cdot p_{i,r} \leq t(K)$$

    Hence, for any subset of resources $J \subset \Resources$, $$\sum_{r\in J} \sum_{i\in K} \sigma_{K,i} \cdot p_{i,r} \leq t(K) \cdot |J|$$

    Thus,
    \begin{equation}
    \forall J \subset \Resources,\ \frac{\sum_{r\in J} \sum_{i\in K} \sigma_{K,i} \cdot p_{i,r}}{|J|} \leq t(K) \label{eq:sum}
    \end{equation}
    Notice that this is an equality when $J$ is a subset of a saturated set $\mathcal{S}$ of any optimal placement $(p_{i,r})_i$. Indeed, by the definition of the saturated set (definition~\ref{def:saturating}) we have that for each $r_s\in\mathcal{S}$, $t(K) = \sum_{i\in K} \sigma_{K,i} \cdot p_{i,r}$.
    \bigskip
    % Note: \rho(\overline{r}_J) does not make any sense now (= 1 in the new formalism)
    We will now prove that $\overline{t}(K) \leq t(K)$. Consider any combined port $\overline{r}_J \in \overline{\Resources}$ whose throughput is $\rho(\overline{r}_J)=|J|$. For any $\overline{r}_J$, by the definition of the dual:
    \begin{align*}
        \sum_{i\in K} \sigma_{K,i} \cdot \rho^{N}_{i,\overline{r}_J} & = \sum_{i\in K} \sigma_{K,i} \cdot \frac{ \delta_{(v_i,\overline{r}_J) \in \overline{E}} }{ |J| } \\
         & = \sum_{i\in K} \sigma_{K,i} \cdot \frac{ \delta_{\{r, (v_i,r)\in E\} \subseteq J} }{ |J| }
    \end{align*}
    where $\delta_{stat} = 1$ if the statement $stat$ is true, and $\delta_{stat} = 0$ if it is false.

    Now, let us show that for any assignment $(p_{i,r})_{i,r}$ in the disjunctive graph, we have:
    \begin{equation}
        \delta_{\{r, (v_i,r)\in E\} \subseteq J} \leq \sum_{r\in J} p_{i,r}
        \label{eq:delta_thm3_1}
    \end{equation}
    In order to prove this statement, we consider the two cases on the value of $\delta$:
    \begin{itemize}
        \item If $\{r, (v_i,r)\in E\} \not\subseteq J$, then the $\delta=0$. Because $p_{i,r}$ are positive values by definition, this inequality is trivially satisfied.

        \item If $\{r, (v_i,r)\in E\} \subseteq J$, then all the neighbors of $r$ in the disjunctive graph are in $J$. Thus, $R_i(p) \subseteq \{r, (v_i,r)\in E\} \subseteq J$. So, $\sum_{r\in J} p_{i,r}=1$. Notice that in this case, we have an equality.
    \end{itemize}

    Therefore, for any $\overline{r}_J \in \overline{\Resources}$:
    \begin{align*}
        \sum_{i\in K} \sigma_{K,i} \cdot \rho^{N}_{i,\overline{r}_J} & \leq \sum_{i\in K} \sigma_{K,i} \cdot \frac{ \sum_{r\in J} p_{i,r} }{ |J| }\\
        & \leq \frac{ \sum_{r\in J} \sum_{i\in K} \sigma_{K,i} \cdot p_{i,r} }{ |J| }
    \end{align*}

    By using equation (\ref{eq:sum}), we have $\overline{t}(K) \leq t(K)$.

    (ii) Now, assuming that $\mathcal r$ is not limited to a few subsets of $\Resources$, let us prove that $\overline{t}(K) = t(K)$.

    For a given optimal assignment $(p_{i,r})_{i,r}$, let us pick $J=\mathcal{S}$, his saturated set of minimal size. Let us show that for this particular $J$, the previously considered inequalities are equalities.

    As mentioned previously, equation (\ref{eq:sum}) is an equality when $J$ is a subset of the saturated set $\mathcal{S}$. Thus, we only need to show that the inequality (\ref{eq:delta_thm3_1}) is an equality for this $J$.

    Notice that for any instruction $v_i$, if we have an edge $(v_i,r) \in E$ when $r\in J=\mathcal{S}$, then by Corollary~\ref{cor:saturatingA}, we have $\{r, (v_i,r)\in E\} \subseteq J$. Thus, given a $v_i$ we have two situations:
    \begin{itemize}
        \item Either there are no edge from $v_i$ to any $r\in J$, then $\delta_{\{r, (v_i,r)\in E\} \subseteq J}=0$, and $\sum_{r\in J} p_{i,r} = 0$. Thus, we have equality.
        \item Or there are an edge from $v_i$ to a saturated resource $r \in J$. Thus, as mentioned before, $\{r, (v_i,r)\in E\} \subseteq J$ and $\delta_{\{r, (v_i,r)\in E\} \subseteq J} = 1 = \sum_{r\in J} p_{i,r}$. Thus, we also have an equality.
    \end{itemize}

    Therefore, the whole chain of inequality linking $\overline{t}(K)$ to $t(K)$ are equalities. Thus, $\overline{t}(K) = t(K)$.
\end{proof}

We have an equality if $\mathcal{r}$ is the set of all subsets of $\Resources$, whose size is exponential in the number of resources. However, the proof shows that we can restrict ourselves to the saturated set $\mathcal{S}$ of an optimal assignment.

In practice, we build $\mathcal{r}$ by considering the abstract resources that directly correspond to the set of resources that a given \muop can be mapped to in the disjunctive mapping.
Then, we recursively apply this rule: if two abstract resources have a non-empty intersection, we then add their union as a new abstract resource. Intuitively, this new abstract resource introduces a new constraint on the valid assignment in the dual, corresponding to a potential saturation of these resources. We end up with a set containing fewer than 14 elements in our experiments.

  \section{Appendix - Automated Resource Mapping Discovery}
\label{app:ress_map_disc}

This section aims at proving mathematically the algorithms presented in Section~\ref{sec:algos}.

\subsection{Definitions}

We consider a bipartite conjunctive mapping as introduced in Definition~\ref{def:conjunctive}.

\begin{restatable}[Extended bipartite conjunctive mapping]{definition}{defextbipconj}
\label{def:extbipconj}
A bipartite conjunctive port mapping is equivalent to a unique extended form that decouples the use of combined resources as either a consequence of the use of simpler resources, or as the sole use of the combined resource.

Let $(V,\Resources,E)$ be a bipartite conjunctive port mapping. Its extended form is a graph $(V,\Resources,E'\cup B)$ with $B$ the set of \emph{back edges} defined by:
\begin{itemize}
    \item $(r,r')\in B\subset \Resources^2 \Leftrightarrow \forall i,\ w_{i,r'} \geq w_{i,r} \land \rho(r') > \rho(r)$.\\
    Then, $w_{r,r'}=1$, and $(r,r') \in B$ is said to be a \emph{back edge}.
    \item $\forall (v,r), (v,r, w_{v,r}') \in E' \Leftrightarrow  (v,r, w_{v,r}) \in E$ with weight $w_{v,r}' = w_{v,r} - \sum_{r'\neq r} w_{v,r}\cdot w_{r,r'} > 0$.\\
    If $w_{v,r}'$ is reduced to $0$, the edge is not in $E'$.
\end{itemize}

An illustrative example is given in figure \ref{fig:extended_from}.
\end{restatable}

\begin{restatable}[Resource usage]{definition}{defresusage}
\label{def:resusage}
Given a bipartite conjunctive port mapping and a microkernel K, we note $w_{K,r}$ the use of the resource $r$ during the execution of K, i.e.
\[w_{K,r}=\sum_{v_i\in K} w_{v_i,r}\]
\end{restatable}

\begin{restatable}[Load of a resource]{definition}{defload}
\label{def:load}
Given the conjunctive port mapping $(V,\Resources,E)$ on the extended form and a microkernel $K$, we note $\load(r)$ the normalized use of the resource $r$ during the execution of the microkernel of execution time $t_{end}$, i.e.
\[\load(r)=\frac{\sum_{v\in K} w_{v,r}+\sum_{r' \in \Resources} w_{K,r'}w_{r',r}}{t_{end}} \]
\end{restatable}

\begin{restatable}[Normalised resource mapping]{definition}{defnormalressmap}
\label{def:normalressmap}
The normalized version $(V,\Resources,E')$ of a conjunctive (or disjunctive) resource mapping is the semantically equivalent resource mapping $(V,\Resources,E)$  where the throughput of every resource has been normalized to 1, thus decreasing the value of the edges:
\[
w^{E'}_{v,r} = \frac{w^E_{v,r}}{\rho(r)}
\]
In the extended resource mapping form, the values of the back-edges become:
\[
w^{B'}_{r,r'} = \frac{w^{B}_{r,r'}}{\rho(r')}
\]
\end{restatable}

\begin{restatable}[Bounds of the normalized back edges]{lemma}{lembackedge}
\label{lem:backedge}
On a normalized bipartite resource mapping on the extended form:
\[
0 \leq w_{r,r'} \leq \frac{1}{2}
\]
\end{restatable}

\begin{figure}
    \centering
    \begin{subfigure}{0.45\linewidth}
    \includegraphics[width=\linewidth]{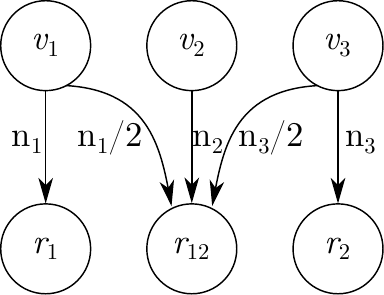}
    \caption{\label{fig:res_example_bi_non_ex}}
    \end{subfigure}
    \hfill
    \begin{subfigure}{0.45\linewidth}
    \includegraphics[width=\linewidth]{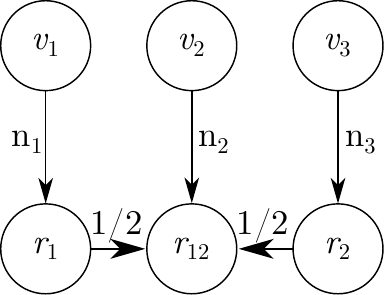}
    \caption{\label{fig:res_example_bi_ex}}
    \end{subfigure}
    \caption{
    \label{fig:extended_from}Conjunctive resource mapping and its extended form; both normalized.}
\end{figure}

\begin{restatable}[$k$-exclusive saturation]{definition}{defkexclsat}
\label{def:kexclsat}
A microkernel $S$ is said to be a $k$-exclusive saturation with $k\in[0,1]$ of a resource $r$ when the maximum of the uses of every resource but $r$ (including the combine resource property) is bounded by $k$, and where $S$ never use a resource that has a back edge with $r$, i.e. when $S$ verifies
\[
\max_{r'\neq r} w_{S,r'} \leq \frac{1 - k}{\ipc{S}}
\]
And
\[
\forall r', w_{r,r'} \neq 0 \Rightarrow w_{S,r'} = 0
\]
We call $S$ an exclusive saturation of $r$ when $S$ is a 1-exclusive saturation of $r$, which means that $S$ only uses the resource $r$.
\end{restatable}

\subsection{Completeness of the Mapping}
\label{App:CompletnessMapping}
We assume that the solver finds an edge $(I, r)$ when we provide a microkernel containing at least once $I$ that saturates $r$. This means that the solver is smart enough to detect that an instruction which uses a resource (and its amount) as long as we provide a benchmark using the instruction limited by it.

\begin{restatable}[Completeness of the output mapping]{theorem}{thmCompletnessSatBenchs}
\label{thm:CompletnessSatBenchs}
Let $r, r_{ST}$ and $r_S$ be three resources, and $S$ and $I$ two microkernels represented as a single vertex combining the use of all their instructions, forming a normalized conjunctive bipartite mapping on the expanded form (see definition~\ref{def:extbipconj}). For the sake of simplicity, we will use greek letters instead of multiple indexes of $\rho$ in this proof:
\begin{itemize}
\item $\ell, \ell'$ (possibly 0) are the back edge from $r_S$ (resp. $r_{ST}$) to $r$ (resp. $r_S$), which  corresponds to $\rho_{r_S,r}$ and $\rho_{r_T,r_S}$, respectively
\item $\alpha, \beta$, and $\gamma$ are used to denote $\rho_{I,r_{ST},}$, $\rho_{I,r_S}$, and $\rho_{I,r}$.
\end{itemize}

Let us assume that $S$ verifies the following properties:
\begin{itemize}
\item $S$ realize a $\frac{1}{4}$-exclusive saturation of $r_S$
\item $S$ uses another resource $r_T$ with a coefficient $\alpha_{ST}$ (noted $\rho_{S,r_T}$ with the former notation)
\item $S$ does not use $r$ apart from the contribution from $r_{ST}$ and $r_T$
\end{itemize}

Then the benchmark $S^{4\cdot\ipc{S}}I^{\ipc{I}}$ saturates $r_S$, thus allowing the solver to find the link $I \to r_S$.
\end{restatable}

\begin{figure}
    \centering
    \begin{subfigure}{0.45\linewidth}
    \includegraphics[width=\linewidth]{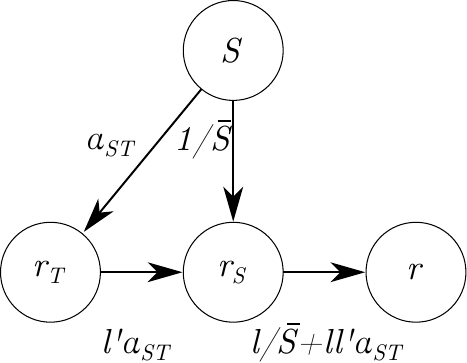}
    \caption{\label{fig:res_sat_proves_S}}
    \end{subfigure}
    \hfil
    \begin{subfigure}{0.45\linewidth}
    \includegraphics[width=\linewidth]{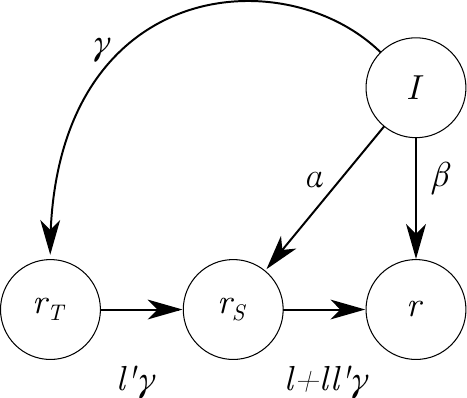}
    \caption{\label{fig:res_sat_proves_I}}
    \end{subfigure}
    \begin{subfigure}{0.45\linewidth}
    \includegraphics[width=\linewidth]{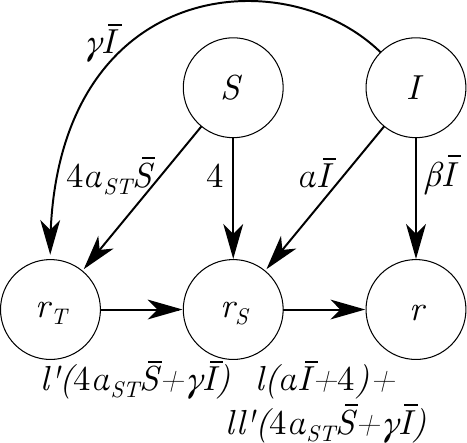}
    \caption{\label{fig:res_sat_proves}}
    \end{subfigure}
    \caption{
    \label{fig:res_sat_proves_general}Saturating benchmarks $S$ and instructions to analyse $I$: individual uses (\ref{fig:res_sat_proves_I} and \ref{fig:res_sat_proves_S}), and benchmark $S^{4\cdot\ipc{S}}I^{\ipc{I}}$ (\ref{fig:res_sat_proves}).
    }
\end{figure}
\begin{proof}
The graph representing the resource usage for one execution of the benchmark $S$ is represented in figure \ref{fig:res_sat_proves_S} and figure \ref{fig:res_sat_proves_I} for $I$. As $S$ realizes a $\frac{1}{4}$-exclusive saturation of $r_S$, then $\rho_{S,r_S} = 1/\ipc{S}$, so that $\ipc{S}$ repetitions of $S$ are needed to load resource $r_S$ with the value 1.

Let us consider the benchmark $S^{4\cdot\ipc{S}}I^{\ipc{I}}$, illustrated in figure \ref{fig:res_sat_proves}.
The load of $r_T$ is:
\[
    \load(r_T) = 4\alpha_{ST}\ipc{S} + \gamma\ipc{I}
\]
The load of $r_S$ is:
\begin{align*}
    \load(r_S) & = \ell'\cdot\load(r_T) + \alpha\ipc{I} + 4\cdot \frac{\ipc{S}}{\ipc{S}}\\
    & = \ell'\cdot(4\alpha_{ST}\ipc{S} + \gamma \ipc{I}) + \alpha\ipc{I} + 4\\
    & = 4\ell'\alpha_{ST}\ipc{S} + \ell'\gamma \ipc{I} + \alpha\ipc{I} + 4
\end{align*}
And the load of $R$ is:
\begin{align*}
    \load(R) & \le \ell\cdot\load(r_S) + \beta\ipc{I}\\
    & =  4\ell'\ell\alpha_{ST}\ipc{S} + \ell'\ell\gamma\ipc{I} + \ell\alpha\ipc{I}  + 4\ell + \beta\ipc{I}
\end{align*}

By lemma \ref{lem:backedge}, $\ell \leq \dfrac{1}{2}$ and $\ell'\leq\dfrac{1}{2}$.

Then
\[
     \load(r) \leq  4\alpha_{ST}\frac{\ipc{S}}{4} + \gamma\frac{\ipc{I}}{4} + \alpha\frac{\ipc{I}}{2}  + \frac{4}{2} + \beta\ipc{I}
\]
But, by definition of the IPC, $\max \left( \alpha, \beta, \gamma \right) = \dfrac{1}{\ipc{I}}$, so
\begin{align*}
     \load(r) & \leq \alpha_{ST}\ipc{S} + \frac{1}{4} + \frac{1}{2}  + \frac{4}{2} + \frac{1}{2}\\
     & \leq \alpha_{ST}\ipc{S} + \frac{13}{4}
\end{align*}

As $S$ realizes a $\frac{1}{4}$-exclusive saturation of $r_S$, then $\alpha_{ST} \leq \frac{3}{4\cdot\ipc{S}}$, so
\begin{align*}
    \load(r) & \leq \frac{3\cdot \ipc{S}}{4\cdot\ipc{S}} + \frac{13}{4}\\
    & \leq 4
\end{align*}

Similarly,
\begin{align*}
    \load(r_T) & \leq 4\alpha_{ST}\ipc{S} + \gamma\ipc{I} \leq 3 + 1\\
    & \leq 4
\end{align*}
To obtain a lower bound on $\load(r_S)$, we use similar bounds: $\alpha \geq 0$ and $\ell' \geq 0$, so
\begin{align*}
    \load(r_S) & = \ell'\cdot(4\alpha_{ST}\ipc{S} + \gamma \ipc{I}) + \alpha\ipc{I} + 4\\
    & \geq 4 + \ipc{I}\cdot (\ell'\gamma + \alpha)
\end{align*}
So, when $r_S$ is used by $S$, either indirectly by $\gamma > 0$ and $\ell' > 0$ or directly when $\alpha > 0$, then  $\load(r_S) > \load(r_T)$ and $\load(r_S) > \load(r)$. So $r_S$ is the bottleneck, and the solver will find the edge $I \to r_S$.
\end{proof}

Note that this proof still stands when $S$ and $I$ use several resources that indirectly contribute to $r_S$, as it is be equivalent to a bigger value of $\gamma$.

\fi

\end{document}